\documentclass[11pt, letterpaper]{article} 
\usepackage[margin=1in]{geometry}
\usepackage{graphicx}          
\usepackage{epstopdf}
\graphicspath{{Figures/}}
\usepackage{pdfpages}
\usepackage{amsthm}
\usepackage{amssymb}
\usepackage{empheq}
\usepackage{color}
\usepackage{adjustbox}
\usepackage{centernot}
\usepackage{csquotes}
\usepackage{mathtools}
\usepackage{mathabx}
\usepackage{multirow,bigdelim}
\usepackage{tikz}
\usepackage{xcolor}
\usepackage{soul}
\usepackage{amssymb,amsmath}
\usepackage{upgreek}
\usepackage{arydshln}
\usepackage{savesym}
\usepackage{nicefrac}
\usepackage{natbib}
\usepackage{algorithm}
\usepackage{algpseudocode}
\usepackage{etoolbox}
\usepackage[colorlinks,citecolor=blue,urlcolor=blue]{hyperref}
\newtheorem{theorem}{\textbf{Theorem}}

\newtheorem{proposition}{Proposition}

\newcommand{\bb}{\mathbf}
\newcommand{\mc}{\mathcal}
\newcommand{\nn}{\nonumber}

\newcommand{\te}{\text}

\newcommand{\bs}{\boldsymbol}
\newcommand{\Card}{\texttt{Card}}

\allowdisplaybreaks
\title{\textbf{Sparsity-Promoting Optimal Control of Cyber-Physical Systems over Shared Communication Networks}} 
\author{Nandini Negi and Aranya Chakrabortty\thanks{N. Negi and A. Chakrabortty are with the Department of Electrical and Computer Engineering at North Carolina State University.   Email: {\tt \{nnegi, achakra2\}@ncsu.edu}}}    
\date{\vspace{-5ex}}
\begin{document}
\maketitle
\begin{abstract}
\noindent Recent years have seen several new directions in the design of sparse control of cyber-physical systems (CPSs) driven by the objective of reducing communication cost. One common assumption made in these designs is that the communication happens over a dedicated network. For many practical applications, however,  communication must occur over shared networks, leading to two critical design challenges, namely - time-delays in the feedback and fair sharing of bandwidth among users. In this paper, we present a set of sparse $\mc{H}_2$ control designs under these two design constraints. An important aspect of our design is that the delay itself can be a function of sparsity, which leads to an interesting pattern of trade-offs in the $\mc{H}_2$ performance. We present three distinct algorithms. The first algorithm preconditions the assignable bandwidth to the network and produces an initial guess for a stabilizing controller. This is followed by our second algorithm, which sparsifies this controller while simultaneously adapting the feedback delay and optimizing the $\mc{H}_2$ performance using alternating directions method of multipliers (ADMM). The third algorithm extends this approach to a multiple user scenario where optimal number of communication links, whose total sum is fixed, is distributed fairly among users by minimizing the variance of their $\mc{H}_2$ performances. The problem is cast as a difference-of-convex (DC) program with mixed-integer linear program (MILP) constraints. We provide theorems to prove convergence of these algorithms, followed by validation through numerical simulations.
\end{abstract}
\section{Introduction}
Recent achievements in engineering, computer science and communication technology have led to the development of large-scale cyber-physical systems (CPSs) that comprise of multiple dynamical systems interacting over robust communication networks. In the past two decades, significant amount of research has been done on the control-theoretic aspects of CPS, for example, see \citet{antsaklis}, \citet{hespanha} and \citet{sinopoli}. Conventional CPS control designs usually rest on two idealistic assumptions, namely, that there is no limit on communication cost, and that the communication network is dedicated for control. In reality, however, network operators prefer sparse communication topologies for reducing the communication cost, and multiple users share the same communication network leading to feedback delay. Both factors give rise to important concerns on closed-loop stability and the fairness of bandwidth allocation among users. 

\par Motivated by this problem, in this paper, we propose a network control design  that lies at the intersection of three design challenges, namely - delays, sparsity of communication network, and optimal $\mc{H}_2$ performance for a LTI system. Results addressing pairs of these three are known in the literature. For example, at the junction of sparsity and optimal control, sparse linear quadratic regulators (LQR) and $\mc{H}_2$ controllers have been proposed by \citet{mihailo} using alternating directions method of multipliers (ADMM), by \citet{wytock} using proximal Newton method, and by \citet{motee} using rank-constrained convex optimization, among others. The approach has been extended to sparse output feedback in \citet{sparse2} using a variant of Kalman filtering.

\par A vast literature exists on the stability and control of LTI systems with time-delays, starting from the seminal works of \citet{kolmanovskii}, \citet{stability}, \citet{krstic}, and others. These papers use various analytical tools such as Riccati matrix equations, Lyapunov-Krasovskii functionals, and small-gain theorem to derive sufficient conditions for stability of time-delayed systems. Optimal control of such systems has been studied in \citet{kushner}, \citet{ross}, \citet{esfahani}, and more recently in \citet{spectral} and \citet{hale} using spectral discretization of delays.

\par However, when all the three design challenges are combined together, the resulting problem becomes far more difficult to solve due to its non-convex nature as well as due to the complicated interplay between the design factors, constrained by nonlinear matrix inequalities (NMIs) arising from delay stability requirements. The problem becomes more complicated when the delay itself is a function of sparsity as eliminating communication links can free up bandwidth and reduce the per-link delay. More importantly, the interpretation of these challenges when multiple users share the same communication network is also an open question as in this case the network operator must divide the bandwidth in a fair way so that each user has a guaranteed relative $\mc{H}_2$ performance while maintaining closed-loop stability. Classical results on bandwidth allocation such as by \citet{net1} only provide protocols defined for static utility functions. Our problem, in contrast, requires bandwidth sharing for dynamic and sparse $\mc{H}_2$ control.
\par \noindent Summarizing the above points, the main contributions of this paper are as follows:
\begin{itemize}
    \item[-] We present two algorithms by which one can design sparse optimal controllers for LTI systems in the presence of feedback delays. Starting from an infinite bandwidth assumption, Algorithm 1 designs a practical finite value of the network bandwidth in Section \ref{sec:preconditioning} by co-designing the delays and the corresponding stabilizing controller. We derive a set of convex relaxations of NMI based sufficient conditions on delay stability to perform this co-design. Thereafter, we present Algorithm 2 in Section \ref{sec:mainalgo}, which uses ADMM to sparsify this controller while minimizing the $\mc{H}_2$ norm.
    \item[-] We next extend the sparse controller of Section \ref{sec:mainalgo} to the case when multiple sets of such controllers need to be designed for $N$ decoupled systems belonging to $N$ different users, sharing the same communication network, while ensuring that each user gets a \textit{fair} share of communication links. The problem is formulated in terms of minimizing the \textit{variance} of the $\mc{H}_2$ performance of each individual user. Algorithm 3 is presented in Section 6 to solve this as a difference-of-convex program with MILP constraints.
    \item[-] We validate our results using numerical simulations in Sections \ref{sec:simulations} and \ref{sec:simulations2}. The former illustrates the inter-dependencies between delay, sparsity and optimal control. The latter shows the efficiency of our bandwidth sharing algorithm in reaching its global minimum.
\end{itemize}
\par\noindent Some preliminary results on this topic have been report\-ed in our recent conference paper \citet{nandini}. The results in this paper, however, are significantly extensive in comparison. For example, \citet{nandini} only minimized a lower bound for the $\mc{H}_2$ norm of a time-delayed system using convex relaxations of bilinear matrix inequalities. In contrast, in this paper, we take a completely different approach by minimizing the exact $\mc{H}_2$ norm of the delayed system  using spectral discretization. The preconditioning algorithm (Algorithm 1) in Section \ref{sec:preconditioning} and the bandwidth sharing algorithm (Algorithm 3) in Section \ref{sec:coop} are also added as new contributions. 
\par\textbf{Notations:} $\mathbb{R}$, $\mathbb{R}_{+}$ and $\mathbb{R}_{++}$ represent the set of real, non-negative real and positive real numbers respectively. $\mathbb{Z}_{+}$ represents the set of non-negative integers. The set $\mathbb{N}_n$ represents the set of natural numbers till $n$. The set $B \backslash A$ denotes the set of elements in $B$ but not in $A$. The Euclidean norm and Frobenius norm are represented by $\|\cdot\|$ and $\|\cdot\|_{\te{F}}$, respectively. Unless otherwise stated, $I_n$ denotes an identity matrix of order $n$. The symbol $\star$ in the lower triangular part of a matrix represents a symmetric block.
\section{Problem Formulation}
 \label{ProblemFormulation}
 \subsection{CPS Modeling}
 \label{sec:CPSModeling}
We consider an LTI system represented by the following state-space equation:
\begin{equation}
\dot{x}(t)=Ax(t)+ B u(t) + B_w w(t) , \label{delayfreestatespace}
\end{equation} 
where $x \in \mathbb{R}^n$, $u \in \mathbb{R}^m$ and $w\in \mathbb{R}^{p}$ are the state, control input and exogenous input, respectively. $A\in\mathbb{R}^{n\times n}$, $B\in \mathbb{R}^{n\times m}$ and $B_{w} \in \mathbb{R}^{n\times p}$ are the corresponding matrices. Assuming $(A,B_w)$ to be stabilizable and $(A,B)$ to be controllable, our objective is to design a static state feedback controller $u(t)=-Kx(t)$, with  $K\in\mathbb{R}^{m\times n}$, that minimizes the $\mc{H}_2$ norm of \eqref{delayfreestatespace}. This controller will be implemented in a distributed way using a communication network as in Fig. \ref{fig:CPS}, which shows the CPS representation of the closed-loop system.
\begin{figure}[hbtp]
\centering
\includegraphics[scale=0.38]{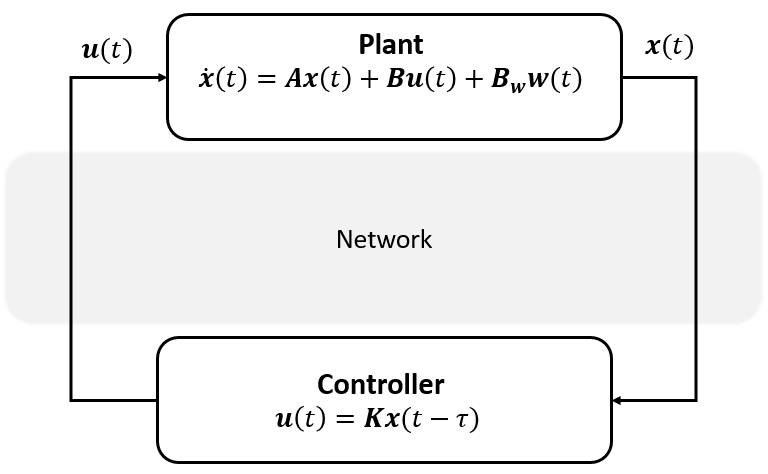}
\caption{Closed-loop CPS representation. Cyber layer receives $x(t)$, computes $u(t)$ after a delay $\tau$ and transmits it to the actuators.}
\label{fig:CPS}
\end{figure}
 The communication of state $x_j$ to compute the input $u_i$ will suffer a delay $\tau_{ij}$ (in seconds), $i\in\mathbb{N}_m$, $j\in\mathbb{N}_n$. This delay consists of two parts, namely, $\tau_{ij}=\tau_{p_{ij}}+\tau_{t_{ij}}$, where $\tau_{p_{ij}}$ is the \textit{propagation delay} and $\tau_{t_{ij}}$ is the \textit{transmission delay}. $\tau_{p_{ij}}$ is defined as the ratio of link length to the speed of light, which is assumed to be a constant denoted as $\tau_p$ for every pair $i,j$. We assume that the communication link connecting any $j$-th sensor to any $i$-th actuator is allotted equal bandwidth, which implies that $\tau_{t_{ij}}$ is also constant for every $i,j$ \citep{ugtext}. We, therefore, denote $\tau_{ij}$ simply as a constant $\tau$. In reality, even if $\tau_{ij}$ deviates from $\tau$ due to traffic fluctuations and network uncertainties, our design will still hold as long as this deviation lies within the stability radius of the plant (See Proposition 1.14 and Theorem 1.16 of \citet{stability2}). The $i$-th control input is written as $u_i (t) =- \sum_{j=1}^{n} K_{ij} x_j(t-\tau)$, and accordingly the closed-loop system is written as:
 \begin{subequations}
\label{delayedstatespaceMain}
\begin{align}
&\dot{x}(t)=Ax(t) - BK x(t-\tau) + B_w w(t), \label{delayedstatespace} \\
&z(t) = C_1 x(t) + D_2u(t), \
C_1=\begin{bmatrix}
Q^{\frac{1}{2}} \\ 0
\end{bmatrix}, \  D_2=\begin{bmatrix}
0 \\ R^{\frac{1}{2}}
\end{bmatrix},\label{delayedstatespace1}
\end{align}
 \end{subequations}
where $Q\succeq 0$, $R\succ 0$. $K$ follows from the minimization of the  $\mc{H}_2$ norm of the transfer function of \eqref{delayedstatespaceMain} from input $w(t)$ to output $z(t)$. In general, $K$ is a dense matrix, which incurs high communication and computation costs. To reduce these costs, we impose a sparsity constraint on $K$. As $K$ is sparsified, the bandwidth $c$ is equally redistributed among the remaining links. From \citet[Ch. 3]{ugtext}, $\tau$ can be written as
\begin{equation}
     \tau=\tau_t+\tau_p =  \underbrace{\kappa \left(\frac{\Card(K)}{c}\right) + \tau_p}_{\mc{Z}\big(\Card(K),c,\tau_p\big)}, \label{Z}
\end{equation}
where $\Card(K)$ denotes the cardinality or the number of non-zero elements in $K$ and $\kappa>0$ is a proportionality constant.The sparsity level $s$ of $K\in\mathbb{R}^{m\times n}$ is defined as $s=m n - \Card (K)$. Equation \eqref{Z} implies that $\tau$ will change as the sparsity level of $K$ changes. This change is captured by the function $\mc{Z}(\cdot)$. 
\subsection{Problem Statement}
Let $\mc{V}_\tau$ be the set of $K$ that stabilize \eqref{delayedstatespaceMain} for a delay $\tau$. Our objective is to design a sparse controller $K$ for the closed-loop stability of \eqref{delayedstatespaceMain} while minimizing its $\mc{H}_2$ norm $J_\tau(K)$. The controller design problem is stated as:
\begin{subequations}
\label{bigh2}
\begin{align}
\boldsymbol{\mc{O}}_1 : \ &\underset{K}{\text{minimize}} \ \ \ \ \ J_\tau(K) +\lambda \Card(K) \label{bigh2:eq1}\\ 
&\text{subject to}   \ \ \ \   K\in \mc{V}_\tau, \label{bigh2:eq2} \\
&\hspace{1.8cm} \ \tau=\mc{Z}\big(\Card (K),c,\tau_p\big), \label{bigh2:eq3}
\end{align}
\end{subequations}
where $\lambda > 0$ is a regularization parameter to emphasize the sparsity of $K$. We assume that the set  $\mc{V}_\tau$ is non-empty. The closed form expression of $J_\tau(K)$ will be derived shortly in Section \ref{DelayedH2norm}, where we present our main algorithm for solving $\bs{\mc{O}}_1$. Before stating that, we first develop a pre-conditioning algorithm by which we can construct an appropriate initial condition for it. 

\section{Preconditioning Algorithm : Co-design of \texorpdfstring{$\tau$}{tau} and \texorpdfstring{$K$}{K} }
\label{sec:preconditioning}
Given $A$, $B$, $B_w$ in \eqref{delayedstatespaceMain}, our sparsity-promot\-ing algorithm for solving $\bs{\mc{O}}_1$ in Section \ref{sec:mainalgo} (namely, Algorithm 2) will require a pair 
\begin{equation}
    (K',\tau') :  K'\in\mc{V}_{\tau'},  \ \tau' =\mc{Z}\big(\Card(K'),c,\tau_p\big), \label{pairrequired}
\end{equation}
as an initial condition, i.e., $(K',\tau')$ must stabilize \eqref{delayedstatespaceMain}, where $\tau'=\mc{Z}\big(\Card(K),c,\tau_p\big)$; $c$ is a {nominal} allotted bandwidth and $\tau_p$ is a known constant. The choice of $c$ plays an important role in finding this pair. For example, let $K_o$ be any matrix such that $(A-BK_o)$ is Hurwitz. Then from \citet[Chapter 2]{stability} it follows that there exists a $\tau^m_o\geq 0$ such that $x=0$ is an asymptotically stable equilibrium of \eqref{delayedstatespace} with $K=K_o$, $w=0$ and $\tau \in [0,\tau^m_o]$; $\tau^m_o$ is referred to as the \textit{delay margin} of $K_o$. Thus, if $c$ is such that $\tau^*_o:=\mc{Z}(\Card(K_o),c,\tau_p)$ is less than or equal to $\tau^m_o$, then one can simply use $(K',\tau')\equiv(K_o, \tau^*_o)$ as an initial condition for Algorithm \ref{algo2}. However, if $c$ is too small so that $\tau^m_o<\tau^*_o$, then a different initial guess needs to be found. In this section we develop an algorithm for finding $(K', \tau')$ for such conservative values of $c$. In case the pair $(K',\tau')$ does not exist for the given $c$, the steps of this algorithm will adapt $c$ to $c'$ so that such a pair can be found. We refer to the interval $[0, \tau^m_o]$ as a \textit{stable delay interval} of $(A, B, K_o)$. Depending on these three matrices, a system may have multiple non-overlapping stable delay intervals. This fact will also be accounted for in the algorithm.
\subsection{Delay Stability Conditions}
As $(K',\tau')$ in \eqref{pairrequired} must stabilize \eqref{delayedstatespaceMain}, we first recall an LMI based sufficient condition that guarantees the stability of \eqref{delayedstatespaceMain} for a given $K$ and $\tau$.
\begin{theorem}
\label{theorem1}
 \citep*{stability} The system in \eqref{delayedstatespace} is asymptotically stable for a given $K$, delay ${\tau}> 0$ and $w=0$ if there exist $n\times n$ matrices $P=P^{T},\ Q_{0}, \ Q_{1}, \  S_{0}=S^{T}_{0}, \ S_{1}=S^{T}_{1}, \ R_{00}=R^{T}_{00}, \ R_{01}, R_{11}=R_{11}^{T}$ such that the following LMIs are feasible:
\begin{subequations}
\label{lmi}
\begin{align}
&\Phi_{\te{L}}=\begin{pmatrix}
\Delta_{0} & \Delta_1 & -D_{1}^a & -D_{1}^{b} \\
\star & S_{1} & -D_{2}^{a} & -D_{2}^{b} \\
\star& \star& \Delta_2 &  0\\
\star & \star & \star  & \Delta_3
\end{pmatrix} \succ 0 , \label{lmi1}\\
&\Psi_\te{L}=\begin{pmatrix}
P & Q_0 & Q_1 \\
\star & R_{00}+\frac{1}{{\tau}} S_{0} & R_{01} \\
\star & \star & R_{11}+\frac{1}{{\tau}} S_{1} 
\end{pmatrix} \succ 0,\label{lmi2}\\
&\Delta_{0} = -PA-A^{T}P	- Q_{0} -Q_{0}^{T} -S_{0}, \\
&\Delta_1 = Q_{1}-PBK, \ \Delta_3 = 3 (S_{0}-S_{1}), \\
& \Delta_2={{\tau}} (R_{00}-R_{11})+(S_{0}-S_{1}),\\
 &  D_1^a = \frac{{\tau}}{2} A^{T}(Q_{0}+Q_{1}) + \frac{{\tau}}{2} (R_{00} +R_{01})-(Q_0-Q_1),\\
& D_2^a = \frac{{\tau}}{2} K^{T} B^T(Q_{0}+Q_{1}) - \frac{{\tau}}{2} (R_{01}^{T} +R_{11}), \\
& D_1^b= -\frac{{\tau}}{2} A^{T}(Q_{0}-Q_{1}) - \frac{{\tau}}{2} (R_{00} -R_{01}),\\
 & D_2^b = -\frac{{\tau}}{2}  K^{T}B^T(Q_{0}-Q_{1}) + \frac{{\tau}}{2} (R_{01}^{T} -R_{11}).
\end{align}
\end{subequations}
\end{theorem}
For our problem, however, the matrices $P$, $Q_0$, $Q_1$, $R_{00}$, $R_{01}$, $R_{11}$, $S_0$, $S_1$ are unknown in addition to the design variables $K$ and $\tau$. As a result, the affine conditions in \eqref{lmi} change to nonlinear matrix inequalities (NMIs). In the following subsections, we will replace the LMI matrices $\Phi_{\te{L}} $ and $\Psi_\te{L} $ in \eqref{lmi} by $\Phi_\te{N} $ and $\Psi_\te{N} $ to signify their conversion to NMIs.
\subsection{Relaxing the NMIs : Design of \texorpdfstring{$K'$}{K'} and \texorpdfstring{$\tau'$}{tau'}}
If the NMIs in \eqref{lmi} could be solved then one would obtain a $(K,\tau)$ with $K\in\mc{V}_{\tau}$, thereby at least partly satisfying \eqref{pairrequired}\footnote{That means that the delay-stability condition is satisfied, but $\tau=\mc{Z}(\Card(K),c,\tau_p)$ is not ensured from the solution $(K,\tau)$ of the NMIs \eqref{lmi1}-\eqref{lmi2}.}. However, solving these NMIs is not tractable. Therefore, we first relax these NMIs to a set of convex constraints. We carry out these relaxations in the form of an iterative algorithm with initial condition $(K_o,\tau^m_o)$ where $K_o$ stabilizes the delay-free system (i.e., $A-BK_o$ is Hurwitz) and $\tau^m_o$ is the delay margin of $(A,B,K_o)$. $\tau^m_o$ can be found from either a frequency-sweeping test as shown in \citet[Theorem 2.6; Example 2.9]{stability}, or by checking the stability of a higher-order Pade approximation of the delayed model. The following theorem relaxes the NMIs \eqref{lmi1}-\eqref{lmi2} into a set of convex constraints to obtain a pair $(K,\tau)$ with $K\in\mc{V}_\tau$ from $(K_o,\tau_o^m)$ such that $\tau > \tau_o^m$.
\begin{theorem}
\label{theorem2}
Let the pair $(K_o,\tau^m_o)$ satisfy the LMIs of Theorem \ref{theorem1} with the auxiliary variables $P_o$, $Q_{0o}$, $Q_{1o}$, $R_{00o}$, $R_{01o}$, $R_{11o}$, $S_{0o}$ and $S_{1o}$. Let us consider perturbations $(\Delta K,\Delta \tau)$ on $(K_o,\tau^m_o)$ and $\Delta P$, $ \Delta Q_{0}$, $ \Delta Q_{1}$, $\Delta R_{00}$, $ \Delta R_{{01}}$, $ \Delta R_{{11}}$, $\Delta S_{0}$ and $\Delta S_{1}$ on the auxiliary variables as the solution of the following constraints:
\begin{subequations}
\label{constraint}
\begin{align}
\Phi_0 \succeq & \ \alpha I, \label{constraint1} \\
\alpha \geq & \ \eta \|\mc{B}_1\| + \gamma (\|\mc{B}_2\| + \|\mc{B}_3\|) + \eta \gamma \|\mc{B}_4\| + \epsilon, \label{constraint2} \\
 \gamma \geq & \  \Delta \tau, \label{constraint3}\\
   \eta \geq & \ \|\Delta K \|, \label{constraint4}\\
\Psi_0 \succeq & \ \beta I, \label{constraint5}  \\
\beta \succeq & \gamma \|\mc{C}_1\| + \epsilon, \label{constraint6}
\end{align}
\end{subequations}
where $\Delta \tau \in\mathbb{R}_+$, $\mc{B}_{i}(K_o,{\tau}^m_o),$ $i\in\mathbb{N}_4$, and $\mc{C}_1(K_o,{\tau}^m_o)$ are given in Appendix (Section \ref{subsec:proofoftheorem1}), $\alpha, \beta\in\mathbb{R}_{++}$ are optimization variables, and $\epsilon, \gamma, \eta\in\mathbb{R}_{++}$ are scalar design variables. Then, the pair $(K=K_o+\Delta K ,\tau= \tau^m_o+\Delta \tau)$ satisfies Theorem \ref{theorem1} with the auxiliary variables $P = P_o+\Delta P$, $Q_0 = Q_{0o} + \Delta Q_0$, $ Q_1 =Q_{1o}+\Delta Q_{1}$, $R_{00} = R_{00o}+\Delta R_{00}$, $R_{01}=R_{01o}+ \Delta R_{01}$, $R_{11}=R_{11o}+ \Delta R_{{11}}$, $S_0 = S_{0o}+ \Delta S_{0}$ and $S_{1} = S_{1o} + \Delta S_{1}$. We use the shorthand notation of $(K,\tau)\in\mc{X}_{(K_o,\tau^m_o)}$ to denote that $(K,\tau)$ is obtained from a given pair $(K_o,\tau^m_o)$ using \eqref{constraint}.
\end{theorem}
\begin{proof}
 The proof is given in Appendix (Section \ref{subsec:proofoftheorem1}).
\end{proof}
\begin{algorithm}[hbtp]
 \caption{Preconditioning Algorithm}
  \label{algo1}
 \begin{algorithmic}[1]
\State \textbf{Input:} Initial feasible point $K_{o}$. Initial Delay $\tau^m_{o}$.
\Statex \hspace{1.1cm} Set $(K_0=K_o,\tau_0=\tau^m_o)$. Obtain Auxiliary
\Statex \hspace{1.1cm} variables $P_0$, $\ldots$, $S_{10}$ from Theorem \ref{theorem1}.
\Statex \hspace{1.1cm} Set $\Delta \tau_0 > \epsilon$ (arbitrary value) and $k=0$.
\While {$\Delta \tau_k>\epsilon$}
\If{$K_k\in \mc{V}_{\tau^*_k}$, where $\tau^*_k=\mc{Z}(\Card(K_k),c,\tau_p)$}
\State End with $(K'=K_k,\tau'=\tau^*_k)$.
\Else
\State \textbf{Input:} $z_0=\tau_k$, $g_0=2(T-\tau_k)$, $D_0=K_k$.
\For{$j=0,1,2,\ldots$}
\State  Use \textbf{$\mu$ finding program} to solve:
\State  $(D_j,\mu_j) = \underset{(D,\mu>0)}{\te{argmin}} \ h_k(z_j + \mu g_j)$  
\State  \hspace{1.5cm} s.t. $ \ \ (z_j+ \mu g_j) \in \mc{X}_{(D_j,z_j)}$.
\State Obtain $z_{j+1} = z_j + \mu_j g_j$
\If{$\|z_{j+1} - \tau_k\| \geq \Delta$}
\State Obtain $z'_{j+1} ={(\Delta - z_j + \tau_k)}/g_j$.
\If{$(D_j,z'_{j+1})$ is stable}
\State \textbf{Stop with }$\tau_{k+1} = z'_{j+1}$.
\EndIf
\EndIf
\If{$\|\nabla h_k(z_{j+1})\| \leq \epsilon$}
\State \textbf{Stop with }$\tau_{k+1} = z_{j+1}$.
\EndIf
\State Obtain $g_{j+1} = -\nabla h_k(z_{j+1}) + \frac{(\|\nabla h_k(z_{j+1})\|^2 g_j)}{\|\nabla h_k(z_{j})\|^2}$.
\EndFor
\State \textbf{Result: $K_k=D_j$, $\tau_k=z_j$}
\EndIf
\State Set $\Delta \tau_{k+1} = \tau_{k+1}-\tau_{k}$. Set $k=k+1$. 
\EndWhile
\State \textbf{Result: $K'=K_k$, $\tau'=\tau_k$}
\end{algorithmic}
 \end{algorithm}
 \begin{algorithm}
 \begin{algorithmic}[1]
  \Statex \textbf{\large $\mu$ finding program : $j$-th iteration}
 \State \textbf{Input:} $z_j$, $g_j$ and $D_j$.
\State \textbf{Input:} Auxiliary variables $P_j$, $\ldots$, $S_{1j}$. 
\State $\bullet$  $P_j$,$\cdots$, $S_{1j}$ known. Solve for $\Delta\tau$,
\State $K=D_j + \Delta K$, $P=P_j+\Delta P$, $\ldots$, $S_1=S_{1j}+\Delta S_1$.
\State \hspace{2cm} $\underset{(K,\tau)}{\te{argmin}} \ h_k(z_j+\Delta\tau)$
\State \hspace{2cm} s.t. \ \    $(K,\tau) \in \mc{X}_{(K_k,\tau_k)}.$
\State $\bullet$ Solve for auxiliary variables $P$, $\ldots$, $S_{1}$
\State \ \ \ \ \ \  minimize \ $ \|P_k-P\|^2 + \cdots + \|S_{1k}-S_{1}\|^2 $
\State \ \ \ \ \ \ \ \ \ \  s.t. \ \ \ \ \  \ $\Phi_L(K,\tau)\succeq \epsilon I, \  \Psi_L(\tau)\succeq \epsilon I$
\State \textbf{Result:} $\mu=(\tau-z_j)/g_j$, $K_{k+1}=K$,  
\State \hspace{1.1cm} $P_{k+1}=P$, $\ldots$, $S_{1,k+1} = S_{1}$. 
\end{algorithmic}
 \end{algorithm}
\par\noindent We next use Theorem \ref{theorem2} to develop an algorithm starting from $(K_o,\tau^m_o)$, where $\Delta K$ and $\Delta \tau$ are evaluated iteratively such that we eventually obtain a stabilizing pair $(K',\tau')$ that fully satisfies \eqref{pairrequired}, i.e., $\tau'=\mc{Z}(\Card(K'),c,\tau_p)$. Equations \eqref{constraint1}-\eqref{constraint6} in Theorem \ref{theorem2} provide the set of constraints that are required to search for a new stabilizing pair $(K_{k+1}=K_k+\Delta K_k,$ $\tau_{k+1} = \tau_k + \Delta \tau_k)$ in the $(k+1)$-th iteration of this algorithm. We incorporate these constraints using the conjugate gradient Steighaug's method \citep{steihaug} as shown next.
\par Let us define $f(\tau)=(T-\tau)^2$, where $T$ is a constant. Starting from the known stabilizing pair $(K_k,\tau_k)$ for \eqref{delayedstatespaceMain}, $T$ can be chosen as $T\geq \mc{Z}(\Card(K_k),c,\tau_p)$ and the $(k+1)$-th iteration of the algorithm can be obtained as the solution of the following program:
\begin{subequations}
\label{optimalgo1}
\begin{align}
\boldsymbol{\mc{O}}_2 :  \ & \underset{(K,\tau)}{\textbf{minimize}}  \ \  \ h_k(\tau)= \ f({\tau}_k) + \nabla f({\tau}_k)^T (\tau-\tau_k)\nn\\[-10pt]
&\hspace{3.7cm}+ \|\tau-{\tau}_k\|^2\\[5pt]
&\textbf{subject to}  \hspace{1cm} \|\tau-{\tau}_k\| \leq \Delta,\\
&  \hspace{2.6cm}  (K,\tau)\in \mc{X}_{(K_k,{\tau}_k)},
\end{align}
\end{subequations}
where $\Delta$ is the radius of the trust-region, and  $\mc{X}_{(K_k,{\tau}_k)}$ is a subset of the stabilizing pairs from Theorem \ref{theorem2}. Algorithm \ref{algo1} describes the iterative solution for \eqref{optimalgo1}. Details on the practical choice of $\epsilon$, $\eta$ and $\gamma$ can be found in \citet{nandini}. The following Theorem provides the conditions under which Algorithm \ref{algo1} converges.
\begin{theorem}
Using Algorithm \ref{algo1}, a sequence $\{\Delta {\tau}_k\} \subset \mathbb{R}_{+}$ is established in $k$ iterations with $\lim_{k\to\infty} \Delta{\tau}_{k} = 0$ if $\eta$ and $\gamma$ are chosen such that the convex program in Theorem \ref{theorem2} is feasible for iterations $\{1,2,\ldots,k\}$.
\end{theorem}
\begin{proof}
 Constraint set in \eqref{optimalgo1} satisfies linear constraint qualification. Additionally, $(K'_{i+1},{\tau}'_{i+1})\in \mc{X}_{(K'_{i},{\tau}'_{i})}$ and $\Delta \tau_i \in \mathbb{R}_{+}$ is ensured for every $i$-th iteration through \eqref{constraint} and the assumption on the values of $\eta$ and $\gamma$. Therefore, using \citet[Theorem 3.3]{steihaug}, global convergence of \eqref{algo1} is guaranteed.
\end{proof}
\par\noindent If at any $k$-th iteration, we obtain a pair $(K_k,\tau_k)$ such that $K_k$ stabilizes \eqref{delayedstatespaceMain} for $\tau^*_k=\mc{Z}(\Card(K_k),c,$ $\tau_p)$, then the test in \eqref{pairrequired} is passed. In that case, we stop Algorithm \ref{algo1}, set $(K',\tau')\equiv (K_k,\tau^*_k)$ and allocate a final bandwidth as $c_f=c$. However, depending on where we start from, Algorithm 1 may converge to a pair $\lim_{k\to\infty}(K_k,\tau_k)$ such that $K_k$ stabilizes \eqref{delayedstatespaceMain} for a stable delay interval $[\tau_{k-},\tau_{k+}]$ ($\tau_k$ belongs to this interval) but does not stabilize it for $\tau^*_k=\mc{Z}(\Card(K_k),c,\tau_p)$. In that case, instead of re-running the algorithm with a different initial pair $(K_o,\tau^m_o)$, an easier alternative is to set $(K',\tau')\equiv (K_k,\tau_{k+})$ if $\tau^*_k > \tau_{k+}$ or $(K',\tau')\equiv (K_k,\tau_{k-})$ otherwise, and update $c$ to $c'=\kappa\frac{\Card(K')}{\tau'-\tau_p}$. The final bandwidth is then updated as $c_f=c'$. We next proceed to Algorithm \ref{algo2} with $(K',\tau')$ as the initial condition satisfying \eqref{pairrequired} with bandwidth $c_f$.
\section{Main Algorithm}
\label{sec:mainalgo}
\subsection{\texorpdfstring{$\mc{H}_2$}{H2} norm of the Delayed Model} 
\label{DelayedH2norm}
We use the method of spectral discretization for approximating the time-delayed system \eqref{delayedstatespaceMain} to a finite dimensional LTI system, as derived in \citep{spectral}. $\mc{H}_2$ norm of the original time-delayed system is then approximated as that of this LTI system. Considering any stabilizing tuple $(K,\tau)$, the interval $[-{\tau},0]$ is first divided into a grid of $\mc{N}$ Chebyshev extremal points
\begin{subequations}
\label{theta}
\begin{align}
&\begin{bmatrix}
\theta_1=-{\tau},  & \cdots,  &\theta_\mc{N} = 0
\end{bmatrix}^T, \ \mathcal{M}_i=\cos\left(\pi - \frac{i \pi}{\mc{N}-1}\right), \\ &\mathcal{\theta}_{i+1}=\frac{{\tau}}{2} ( \mathcal{M}_i - 1), \ i \in \{0,1,\ldots,\mc{N}-1\}.
\end{align}
\end{subequations}
We define function segments $\phi(\theta) = x(t+\theta) , \ -{\tau} \leq \theta \leq 0$, which are used to represent both the past (tail) state  $\phi(\theta_i)$, $i\in\mathbb{N}_{\mc{N}-1}$ and the present (head) state $x(t)$ in the augmented state vector $\zeta\in\mathbb{R}^{\mc{N}n}$ as
\begin{equation}
\zeta = \left[\underbrace{\phi^T(-\tau), \ \phi^T(\theta_2), \cdots, \ \phi^T(\theta_{N-1})}_{tail},  \ \underbrace{x^T(t)}_{head} \right]^T.
\end{equation}
Let $\mc{W}:=\mc{L}_2([-\tau,0],\mathbb{R}^n)\times \mathbb{R}^n)$ be defined as the vector space containing $\phi(\theta)$, $-\tau\leq\theta\leq 0$. The discretization in \eqref{theta} replaces the infinite dimensional space $\mc{W}$ by $\mathbb{W}$, which is the space of discrete functions over the $\mc{N}$ points in \eqref{theta}. A linear operator $\mc{A}$ is defined to act on $\zeta=(\phi,x)\in\mathbb{W}$ such that it mimics the behavior of the delayed system \eqref{delayedstatespaceMain} in the following way:
\begin{equation}
\mathcal{A}\zeta = \mathcal{A}\begin{bmatrix}
\phi(\theta_1=-\tau) \\ \vdots \\ x(t)
\end{bmatrix} = \begin{bmatrix}
\partial_{\theta_1} \phi(\theta_1=-\tau) \\\vdots \\ A x(t) + K \phi(\theta_1=-\tau)
\end{bmatrix}.
\end{equation}
Let us define matrices $M_i$, $i\in\mathbb{N}_\mc{N}$ as follows: 
\begin{equation}
M_i=\left[I_n e_{1}, \ \cdots, \ I_n e_{j=i}, \  \cdots, \ I_n e_{\mc{N}} \right]^T, \label{Mdef} 
\end{equation}
where $e_j = 1$ if $j=i$ and $0$ otherwise. The approximated LTI system can then be written as
\begin{align}
\label{augmentedsystem}
\dot{\zeta}(t) = \mc{A}\zeta(t) + \mc{B} w(t),
\end{align}
where $\mc{B}=M_\mc{N} B_w \in \mathbb{R}^{\mc{N}n\times p}$ is the exogenous input matrix for this augmented system. $\mc{A} \in\mathbb{R}^{\mc{N}n\times \mc{N}n}$ represents the closed-loop state matrix, which can be further decomposed as
\begin{align}
\mc{A} = \tilde{A}(\tau) - M_\mc{N} B K M^T_1, \label{closedloopA}
\end{align}
where $\tilde{A}(\tau)$ is defined using \eqref{theta} as
\begin{align}
&\tilde{A}_{ij} =\begin{cases}
q_j(\theta_i) I_n, \ j\in\mathbb{N}_\mc{N}, \ i\in\mathbb{N}_{\mc{N}-1}\\
A, \ \hspace{1.1cm} i=j=\mc{N}\\
0, \ \hspace{1.1cm} \text{else},
\end{cases}\nn\\
&q_i(\theta_i) = \frac{1}{\theta_{ji}}\prod_{m=1,\ m\neq i, j}^\mc{N} \frac{\theta_{im}}{\theta_{jm} }, \ q_j(\theta_i) = \sum_{m=1, \ m\neq j}^\mc{N} \frac{1}{\theta_{jm}},
\end{align}
and $\theta_{pq}=\theta_p - \theta_q$. The $\mc{H}_2$ norm of the time-delayed model can finally be approximated by the $\mc{H}_2$ norm of \eqref{augmentedsystem} which is defined as follows:
\begin{subequations}
\label{Jtauexpression}
\begin{align}
&J_\tau(K,\tau) = \text{Trace} (\mathcal{C} \mc{L}\mathcal{C}^T) = \text{Trace}(\mathcal{B}^T \mc{P}\mathcal{B})\label{H2tau}\\
&\mc{A}(K,\tau) \mc{L} + \mc{L} \mc{A}^T(K,\tau) = -\mathcal{B} \mathcal{B}^T, \label{lyapL} \\
&\mc{A}^T(K,\tau) \mc{P}  + \mc{P}\mc{A}(K,\tau) = - \mathcal{C}^T \mathcal{C},\label{lyapP}\\
&\mathcal{C}^T \mathcal{C} =\tilde{Q} + M_1 K^T R K M^T_1,
\end{align}
\end{subequations} 
where $\tilde{Q}=\texttt{BlkDiag}\left(\bs{0},Q\right)$, and $\mc{P},\mc{L} \in \mathbb{R}^{\mc{N}n \times \mc{N}n}$ are positive-definite Lyapunov matrices. The $\mc{H}_2$ norm $J_\tau$ is finite and positive for a stable delayed system. We next use this expression as the objective function in the ADMM formulation for solving $\bs{\mc{O}}_1$.
\subsection{Design Formulation using ADMM}
As $\bs{\mc{O}}_1$ has two conflicting objectives, namely, sparsity promotion and minimization of the $\mc{H}_2$ norm, it fits the structure of ADMM-based optimization. \citet{mihailo} formulated an ADMM based design for the delay-free LTI system. In our problem, we include the added objective of maintaining stability of \eqref{delayedstatespace} while updating the delay as a function of sparsity. We relax $\bs{\mc{O}}_1$ in the ADMM form as follows:
\begin{subequations}
\label{H2new}
\begin{align}
\boldsymbol{\mc{O}}_3 : \ &\underset{K,G}{\text{minimize}}  \ \ J_\tau(K,\tau) + \lambda g(G)  \label{H21new}\\ 
& \ \ \  \ \ \ \text{s.t.} \ \ \  \  K-G=0 , \\
& \hspace{1.5cm}\tau = \mc{Z}\left(\Card(K),c_f,\tau_p\right), \label{addmconstraint}
\end{align}
\end{subequations}
where $c_f$ and $\tau_p$ are known constants, and $\lambda$ is a known regularization parameter. Due to the non-convex nature of $\Card(K)$, sparsity is promoted through its closest convex approximation, which is the weighted $l_1$ norm $g(G)$. The norm is computed iteratively as
\begin{equation}
\label{wl1}
g(G_{i+1})=\sum_{p,q} W_{p,q} |{G_{i+1}}_{p,q}|, \  W_{p,q}=(|{G_{i}}_{p,q}| +\epsilon_{1})^{-1}, \end{equation}
where $G_{i}$ is known from the $i$-th iteration and $0<\epsilon_1 \ll 1$. $\bb{\mc{O}}_3$ is solved using two loops. The inner ADMM loop will update $K$ and $G$ assuming $\tau$ to be constant while the outer loop will update $\tau$ according to \eqref{addmconstraint}, as shown shortly in Section \ref{sec:changingdelaywithsparsity}. The augmented Lagrangian for the inner ADMM loop can be written as:
\begin{equation}
\mc{L}_{p}=J_\tau(K) + \lambda g(G) + \text{Tr}\big(\Lambda^{T}(K-G)\big) + \frac{\rho}{2}\|K-G\|^{2}_{\te{F}},
\end{equation}
where $\Lambda \in \mathbb{R}^{m\times n}$ is the dual variable and $\rho$ is a constant. $\lambda$ is increased in steps to allow gradual sparsity promotion. For the rest of this subsection, we continue with the formulation of the inner-loop. The outer-loop for $\tau$-update is presented in the next subsection.
\par\noindent At any $(k+1)$-th iteration of the inner loop, the algorithm alternates between the following three minimization problems:
\par\noindent $\bullet$ $K$-\textbf{min} : Minimizes the $\mc{H}_2$ norm of \eqref{delayedstatespaceMain} as
\begin{equation}
\label{kmin}
K_{k+1} = \underset{K_{k+1}}{\te{argmin}} \ \ J_{\tau}(K_{k+1}) + \frac{\rho}{2}\|K_{k+1}-G_{k}+\frac{1}{\rho}\Lambda_{k}\|^2_{\te{F}}.
\end{equation}
$\bullet$ $G$-\textbf{min} : Minimizes the weighted $l_1$ norm as:
\begin{equation}
\label{gmin}
G_{k+1} = \underset{G_{k+1}}{\te{argmin}} \ \ \lambda g(G_{k+1}) + \frac{\rho}{2}\|G_{k+1}-K_{k+1}-\frac{1}{\rho}\Lambda_{k}\|_{\te{F}}.
\end{equation}
$\bullet$ $\Lambda$-\textbf{Update} : $\Lambda_{k+1} = \Lambda_k + (K_{k+1} - G_{k+1})$.

\par\noindent The $K$-min step can be solved using Anderson-Moore method \citep{lqrgrad}. This method requires the gradient of $J_\tau$, which is derived as follows.
\begin{theorem}
\label{theorem4}
The gradient of $J_\tau(K)$ is given as:
\begin{equation}
\nabla J_\tau (K) = 2 ( R K M^T_\mc{N} \mc{L}(K) M_\mc{N} - B^T M^T_\mc{N} \mc{P}(K) \mc{L}(K) M_1 ), \label{grad}
\end{equation}
where $\mc{L}(K)$ and $\mc{P}(K)$ are obtained from \eqref{lyapL} and \eqref{lyapP}, respectively.
\end{theorem}
\begin{proof}
Taking the partial derivative of \eqref{lyapP}, we get
\begin{align}
\mc{A}^T \partial \mc{P}+ \partial \mc{P} \mc{A} =  &M_1 \partial K^T B^T M^T_\mc{N} \mc{P} + \mc{P} M_\mc{N} B \partial K M_1^T\nn\\
 - M_\mc{N}\partial &K^T R K M^T_\mc{N} - M_\mc{N} K^T R \partial K M^T_\mc{N}. \label{eqA}
\end{align}
Partial derivative of the $\mathcal{H}_2$ norm in \eqref{H2tau} is
\begin{subequations}
\label{partial}
\begin{align}
 &\partial J_\tau(K) \ dK =\te{Tr}(\partial \mc{P} \mc{B} \mc{B}^T) = \te{Tr} ( \mc{B}^T \partial \mc{L} \mc{B} ).\label{eqC} \\
&\te{Also,} \ \ \ \partial J_\tau(K) \ dK = \text{Tr}(\nabla J_\tau(K)^T dK). \label{eqB}
\end{align}
\end{subequations}
Multiplying \eqref{lyapP} by $\partial \mc{P}$ and taking its trace, we get
\begin{align}
\te{Tr}\left(\mathcal{B}^T \partial \mc{P} \mathcal{B}\right) = -\te{Tr} \left(\partial \mc{P} \mc{A} \mc{L}   +  \mc{A}^T \partial \mc{P} \mc{L} \right). \label{eqD}
\end{align}
\par Using \eqref{partial} and \eqref{eqD}, we can write
\begin{align}
\text{Tr}(dK^T \nabla J_\tau(K)) = -\te{Tr} \left(\partial \mc{P} \mc{A} \mc{L}   +  \mc{A}^T \partial \mc{P} \mc{L} \right). \label{eqE} 
\end{align}
We next multiply \eqref{eqA} by $\mc{L}$ and take its trace. Using \eqref{eqE}, $\text{Tr}(dK^T \nabla J_\tau(K))$ is equivalent to
\begin{align}
& \te{Tr} (-2 \partial K^T B^T_2 M^T_\mc{N} \mc{P} \mc{L} M_1 + 2 \partial K^T R K M^T_\mc{N} \mc{L} M_\mc{N}).
\end{align}
The final expression for the gradient \eqref{grad} follows from the above equation.
\end{proof}
\par\noindent We next use Theorem \ref{theorem4} to derive the solution of the $K$-min step, followed by the solution of the $G$-min step.
\subsubsection{Solving the \texorpdfstring{$K$}{K}-min Step}
\label{sectionkmin}
To solve the $K$-min step, we employ the Anderson-Moore method using the gradient of the objective function in \eqref{kmin}. Denoting $V_k = G_k - \frac{1}{\rho} \Lambda_k$ and holding $\tau$ constant for this step, we use Theorem \ref{theorem4} to derive the necessary optimality condition for \eqref{kmin} as:
 \begin{align*}
\frac{\rho {R}^{-1}}{2} K^{+} + K^{+} (M^T_N \mc{L} M_N) -R^{-1} \big(B^T M^T_N \mc{P}\mc{L}M_1 + \frac{\rho}{2} V_{k}\big)=0,
\end{align*}
where $K^{+}$ is the descent direction, and $\mc{L}(K_k)$ and $\mc{P}(K_k)$ are the solutions of \eqref{lyapL} and \eqref{lyapP}, respectively. The Armijo-Goldstein rule for determining the step-size guarantees $\lim_{k\to \infty} \|\nabla J_\tau (K_k)\|=0$ \citep[Theorem 4.2]{lqrgrad}.
 \subsubsection{Solving the \texorpdfstring{$G$}{G}-min Step}
\label{sectiongmin}
The solution of the $G$-min step is well-known in the literature, and is given by the soft-thresholding operator \citep{boyd}. Let $K_{pq}$ be the $(p,q)$-th element of the solution of the $K$-min step. Denoting $U_{pq}=\Lambda_{pq} + \rho K_{pq}$, the solution of the $G$-min step is given as
\begin{equation}
G_{pq} = \begin{cases}
\frac{1}{\rho}\left(1-\frac{\lambda W_{pq}}{|U_{pq}|}\right)U_{pq}, \ \ \te{If} \  {|U_{pq}|} > {\lambda} \\
0, \ \ \hspace{2.5cm}   \te{otherwise}.
\end{cases}\label{gmin2}
\end{equation} 
A high value of $\lambda$ favors the promotion of zero entries in $G$. The ADMM loop terminates in $K\in\mc{I}_s$ upon reaching the convergence criteria given in \citet[Section 3.3]{boyd}, where $\mc{I}_s$ represents a set of matrices with a specific sparsity structure. A polishing step will be introduced in Section \ref{sec:polishingstep} for further improving $K$.
\subsection{Update of \texorpdfstring{$\tau$}{tau}}
\label{sec:changingdelaywithsparsity}
\par The delay $\tau$ is assumed to be constant in the inner loop since it is a function of $\Card(K)$ (and not $K$). Moreover, Theorem \ref{theorem4} can guarantee a solution of $K$-min only if $\tau$ remains constant. For these reasons, $\tau$ is updated with respect to $\Card(K)$ (i.e., according to \eqref{addmconstraint}) at the beginning of every re-weighting step \eqref{wl1}, which we refer to as the outer loop. Thus, at every outer loop iteration, sparsity level of $K$ decides the value of $\tau$ and, in turn $\tau$ affects the future value of sparse $K$ through the three steps of ADMM. This is another difference between our proposed algorithm and that in \citet{mihailo}.
 \par Before proceeding with the outer loop design, we first recall a fundamental property from the stability theory of delayed systems following \citet[Chapter 2]{stability}. The characteristic polynomial for \eqref{delayedstatespace} with $w=0$ is
\begin{equation}
 p(j\omega;e^{-j\tau \omega})=\te{det}\left( j\omega I-A+BKe^{-j\tau \omega}\right) . \label{characteristicroots}
\end{equation}
If $\tau^m$ is the delay margin of $(A,B,K)$ in \eqref{delayedstatespace}, then the frequency $\omega$ at which $p(j\omega;e^{-j\tau^m \omega})=0$ becomes the first crossing of the characteristic roots of \eqref{characteristicroots} from the stable left hand complex plane to the unstable right half complex plane, also known as a \textit{zero crossing}. For any given system $(A,B,K)$, multiple but finite zero crossings exist, which indicate that \eqref{delayedstatespace} is also stable for some delay intervals other than $[0,\tau^m]$. We assume that $\mathfrak{n}$ such crossings exist, and denote the crossings as $\omega_k$, $k\in\mathbb{N}_\mathfrak{n}$. Also, we define a dimensionless quantity $\Hat{\omega}_k\in[0,2\pi]$ such that $p(j\omega;e^{-j\hat{\omega}_k})=0$ $\forall \ \omega > 0$. Defining $\nu_k = \nicefrac{\Hat{\omega}_k}{\omega_k}$, we assume without loss of generality that
\begin{equation}
    \nu_0=0<\nu_1=\tau^m < \nu_2 < \cdots < \nu_\mathfrak{n}, \label{nu}
\end{equation}
where $\nu_k$, $k\in\mathbb{N}_{\mathfrak{n}}$. The system in \eqref{delayedstatespace} is stable for any $\tau \in (\nu_k,\nu_{k+1})$ if it is known to be stable for some $\tau^*$ in that interval. As mentioned in Section \ref{sec:preconditioning}, such continuous intervals are called the \textit{stable delay intervals} of  $(A,B,K)$. We next derive the update law for $\tau$ using the concept of stable delay intervals.  
\par\noindent Let the $i$-th re-weighting step start with the stabilizing pair $(K_{i},\tau_{i})$ and end with the stabilizing pair $(K_{i+1},\tau_{i})$, where $\Card (K_{i+1})\leq \Card (K_i)$. In the next iteration, $\tau_i$ should be updated to $\mc{Z}(\Card(K_{i+1}),c_f,\tau_p)$, which is denoted by $\tau^*_{i+1}$. The following scenarios can arise depending on whether the pair $(K_{i+1},\tau^*_{i+1})$ is stabilizing or not:
\begin{description}
\item[4-1] If $(K_{i+1},\tau^*_{i+1})$ stabilizes \eqref{delayedstatespace}, we set $\tau_{i+1}=\tau^*_{i+1}$ and start the ADMM loop. We illustrate this case in \textbf{4-1} panel of Fig. \ref{fig:changingdelay}.
\item[4-2]  If $(K_{i+1},\tau^*_{i+1})$ \textit{does not} stabilize \eqref{delayedstatespace}, then $\tau_{i+1}$ cannot be updated to $\tau^*_{i+1}$ since it does not belong to any stable delay interval of $(A,B,K_{i+1})$. To ensure that the bandwidth does not increase, we look for the closest stable delay intervals towards the \textit{right} of $\tau^*_{i+1}$ as shown in panels \textbf{4-2a} and \textbf{4-2b} of Fig. \ref{fig:changingdelay}. Let us denote the stable interval in which $\tau_i$ belongs to as $[\nu_k,\nu_{k+1}]$, where $k\in\mathbb{N}_\mathfrak{n}$. Two situations can arise in this case. The first situation (panel \textbf{4-2a}) is when the closest stable delay interval on the right of $\tau^*_{i+1}$ is $[\nu_k,\nu_{k+1}]$. We denote the start of this interval by $\nu^*$, and update $\tau_{i+1} =\nu^*$. In the second situation (panel \textbf{4-2b}), the closest stable delay interval to the right of $\tau^*_{i+1}$ is $[\nu_r,\nu_{r+1}]$, where $r<k$. In this case, assuming that one can find this stable interval, we denote $\nu^*=\nu_r$ and set $\tau_{i+1}=\nu^*$. In reality, however, for this case it may be difficult to find the exact location of the nearest interval $[\nu_r,\nu_{r+1}]$ numerically. Since the location of $[\nu_k,\nu_{k+1}]$ is already known from the $i$-th iteration, $\tau_{i+1}$ can be updated to $\nu^*=\nu_k$ for this case as well.
 \end{description}
 \begin{figure}[hbtp]
      \centering
      \includegraphics[scale=0.465]{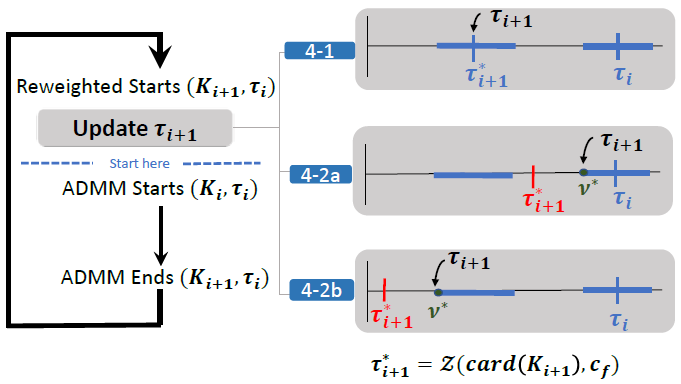}
      \caption{Updating $\tau_{i+1}$ for the $(i+1)$-th reweighted step}
      \label{fig:changingdelay}
  \end{figure}
   \begin{algorithm}[bt]
 \caption{Main Algorithm}
 \label{algo2}
 \begin{algorithmic}[1]
\State \textbf{Input:} Initial feasible point $K'$. Initial delay $\tau'$ such
\Statex that $\tau'= \mc{Z}(\Card(K'),c_f,\tau_p)$.
\For {$\lambda$=$\lambda_i$}
\State \textbf{Set : } $K_0=K'$. 
 \For{Re-weighting Step $i=0$ to $i=r_{max}-1$}
   \State Evaluate $\tau^*_i=\mc{Z}\left(\Card (K_i),c_f,\tau_p\right)$.
   \If{$K_i\in\mc{V}_{\tau^*_i}$}
      \State Update $\tau_i=\tau^*_i$ (\textbf{4-1} in Section \ref{sec:changingdelaywithsparsity}).
   \Else
      \State Update $\tau_i$ according to \textbf{4-2} in Section \ref{sec:changingdelaywithsparsity}.
       \EndIf
      \State \textbf{Input : $(K_i,\tau_i)$}
    \While{ADMM Criteria not satisfied}
         \State $\bullet$ Solve $K$-min in \eqref{kmin} using Section \ref{sectionkmin}.
         \State $\bullet$ Solve $G$-min in \eqref{gmin} using Section \ref{sectiongmin}.
         \State $\bullet$ Update $\Lambda$ using $\Lambda$-update step.
\EndWhile  
\State \textbf{Result : }ADMM ends with $(K_{i+1},\tau_i)$.
\State Carry out the Polishing Step in Section \ref{sec:polishingstep} on
\Statex \hspace{1cm} $(K_{i+1},\tau_i)$ to get $(K'_{i+1},\tau_i)$.
\State Update $W$ from \eqref{wl1}.
\EndFor
\State \textbf{Result : } $K_{i+1}=K'_{i+1}$. 
\EndFor
\State \textbf{Result : } Stabilizing pair $(K_f,\tau_f)\equiv (K_{r_{max}},\tau_{r_{max}})$, where $\tau_f\geq \mc{Z}(\Card(K_f),c_f,\tau_p)$.
 \end{algorithmic}
 \end{algorithm}

\par\noindent Let $(K_f,\tau_f)$ denote the output of Algorithm \ref{algo2}. If \textbf{4-1} is encountered then $\tau_f=\mc{Z}(\Card(K_f),c_f,\tau_p)$. If  \textbf{4-2} is encountered, we may obtain $(K_f,\tau_f)$ such that $\tau_f > \mc{Z}(\Card(K_f),c_f,\tau_p)$, which can be accommodated by decreasing the allotted bandwidth $c_f$. While this saves us bandwidth, it may cost us $\mc{H}_2$ performance as we obtain a controller for a higher feedback delay than necessary. Also, in this case it is easy to verify that the difference between $\tau_f$ and $\tau_f^*$ (where, $\tau_f^*=\mc{Z}(\Card(K_f),c_f,\tau_p)$) is always bounded as
\begin{equation}
\tau_f - \tau^*_f \leq \frac{\Card(K_0)-\Card (K_f)}{c_f} -\epsilon r_{max},
\end{equation}
where $\epsilon=\max_i(\tau^*_{i+1}-\tau_{i+1})$, and $r_{max}$ is the total number of re-weighting steps. 
\subsection{Polishing Step}
\label{sec:polishingstep}
The final step of the inner ADMM loop is to refine its output $K\in\mc{I}_s$ to $K'\in\mc{I}_s$ such that the latter produces the least $J_\tau$  attainable for the sparsity set $\mc{I}_s$. The gradient of $J_\tau(K)$ for fixed $\mc{I}_s$ is given by
\begin{equation}
\nabla J_{\tau}|_{\mc{I}_s}=\nabla J_{\tau} \circ \mc{I}_s.
\end{equation}
The ADMM loop of the $(i+1)$-th re-weighting step produces a stabilizing tuple $(K_{i+1}\in \mc{I}_s,\tau_i)$, as shown in line-17 of Algorithm \ref{algo2}. Newton's method can be used to construct the polishing step with a quadratic convergence rate. Let the output of the polishing step be
\begin{align}
&K'_{i+1}=K_{i+1} + \mathbf{s}\tilde{K}, \label{updateK}\\
\te{where} \ &\tilde{K}=-(\nabla^2(J_\tau))^{-1} \nabla J_{\tau}|_{\mc{I}_s}, \label{Ktilde} 
\end{align}
where $\tilde{K}\in \mc{I}_{s}$ is the Newton's descent direction and $\mathbf{s}$ is the appropriately chosen step-size. Using \citet{lqrgrad}, \eqref{Ktilde} is written in the form of the following necessary optimality condition:
\begin{align}
\resizebox{0.9\linewidth}{!}{$2(R\tilde{K} M^T_\mc{N} \mc{L} M_\mc{N} - B^T M^T_\mc{N} Z_2 \mc{L} M_1 + R K_{i+1} M^T_\mc{N} Z_1 M_\mc{N} - B^T M^T_\mc{N} \mc{P} Z_1 M_1 ) \circ \mc{I}_s + \nabla J_{\tau} \circ {\mc{I}_s}=0,$} \label{nece}
\end{align}
where $\mc{L}(K_{i+1})$ and $\mc{P}(K_{i+1})$ are computed from \eqref{lyapL} and \eqref{lyapP}, respectively. The matrices $Z_1=\partial \mc{L}$ and $Z_2=\partial \mc{P}$ are computed by differentiating \eqref{lyapL} and \eqref{lyapP} to obtain the following equations:
\begin{align}
&\mc{A}Z_1 +Z_1 \mc{A}^T -  G_1- G^T_1=0,  \label{eq51}\\
&\mc{A}^T Z_2  +  Z_2 \mc{A} - G_2 - G^T_2  = 0, \label{eq61}
\end{align}
where $\displaystyle G_1=M_\mc{N} B \tilde{K} M^T_1 \mc{L}$, $G_2=\mc{P} M_\mc{N} B \tilde{K} M^T_1 -M_1 K$ $R\tilde{K} M_1^T$ and $\mc{A}$ follows from \eqref{closedloopA} with $K_{i+1}$ as the feedback gain. We first compute the gradient $\nabla J_\tau (K_{i+1})|_{\mc{I}_s}$ and then solve \eqref{nece}, \eqref{eq51} and \eqref{eq61} for the unknown variables $Z_1$, $Z_2$ and $\tilde{K}$. Beginning from $K_{i+1}$, a sequence of decreasing $\{J_\tau(K'_{i+1})\}$ is generated by using Armijo-Goldstein rule to choose step-size $\mathbf{s}$ in \eqref{updateK}. Conjugate gradient method can be used to compute the Newton direction similar to \citet[Section IV-B]{fardad}.
\section{Examples}
\label{sec:simulations}
We illustrate the effectiveness of algorithms 1 and 2 using numerical examples. We refer to the two algorithms jointly as a \textit{delay-aware} algorithm, as compared to the \textit{delay-agnostic} algorithm of \citet{mihailo} where no feedback delay was considered. We present three case studies, each for the following two conditions:
\par \textbf{S1} $\tau$ is constant.
\par \textbf{S2} $\tau$ changes with $\Card(K)$ following \eqref{Z}.
\begin{figure}[hbtp]
\begin{minipage}{0.5\linewidth}
\centering
\includegraphics[scale=0.4]{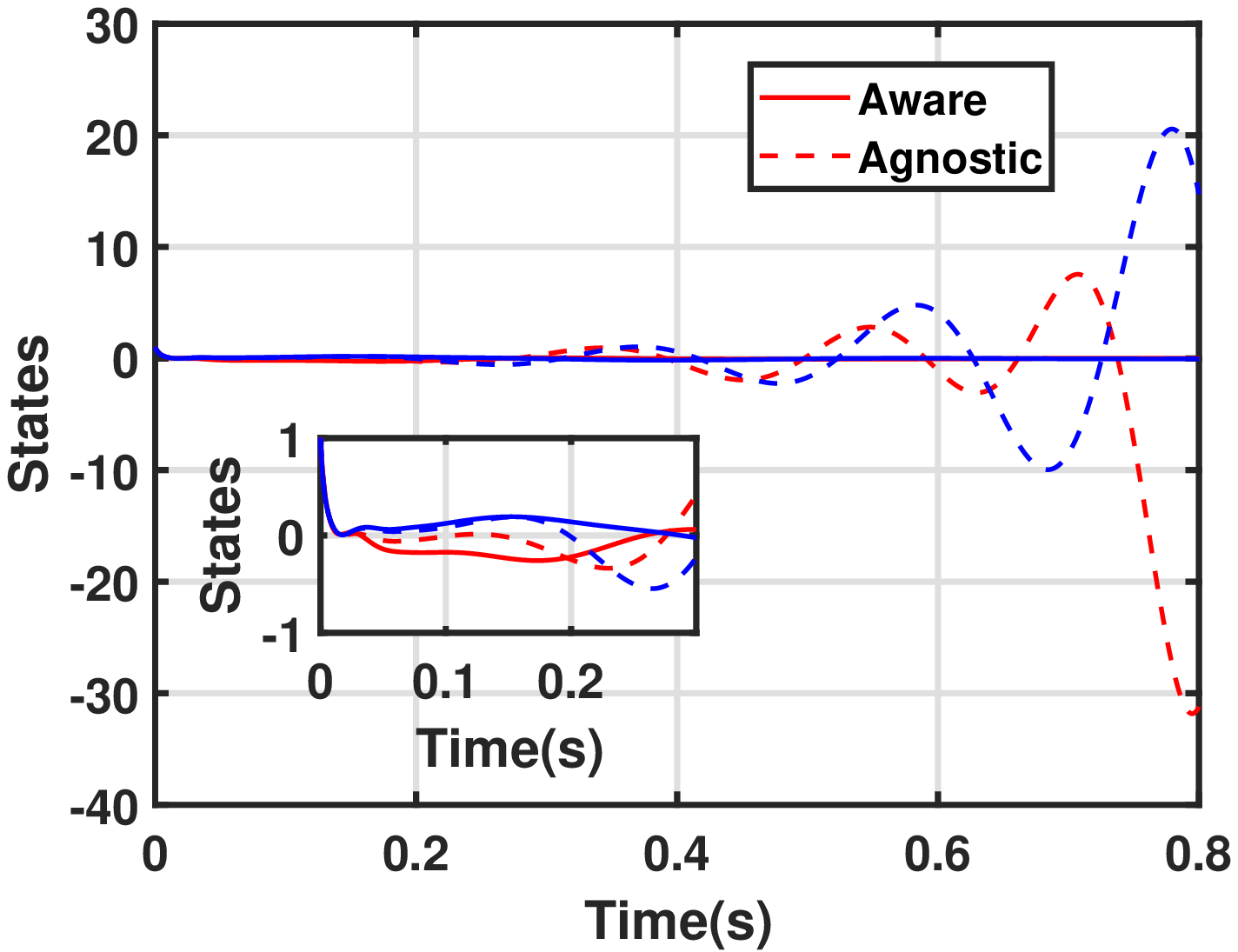}\\
{(a)}
\end{minipage}\begin{minipage}{0.5\linewidth}
\centering
\includegraphics[scale=0.42]{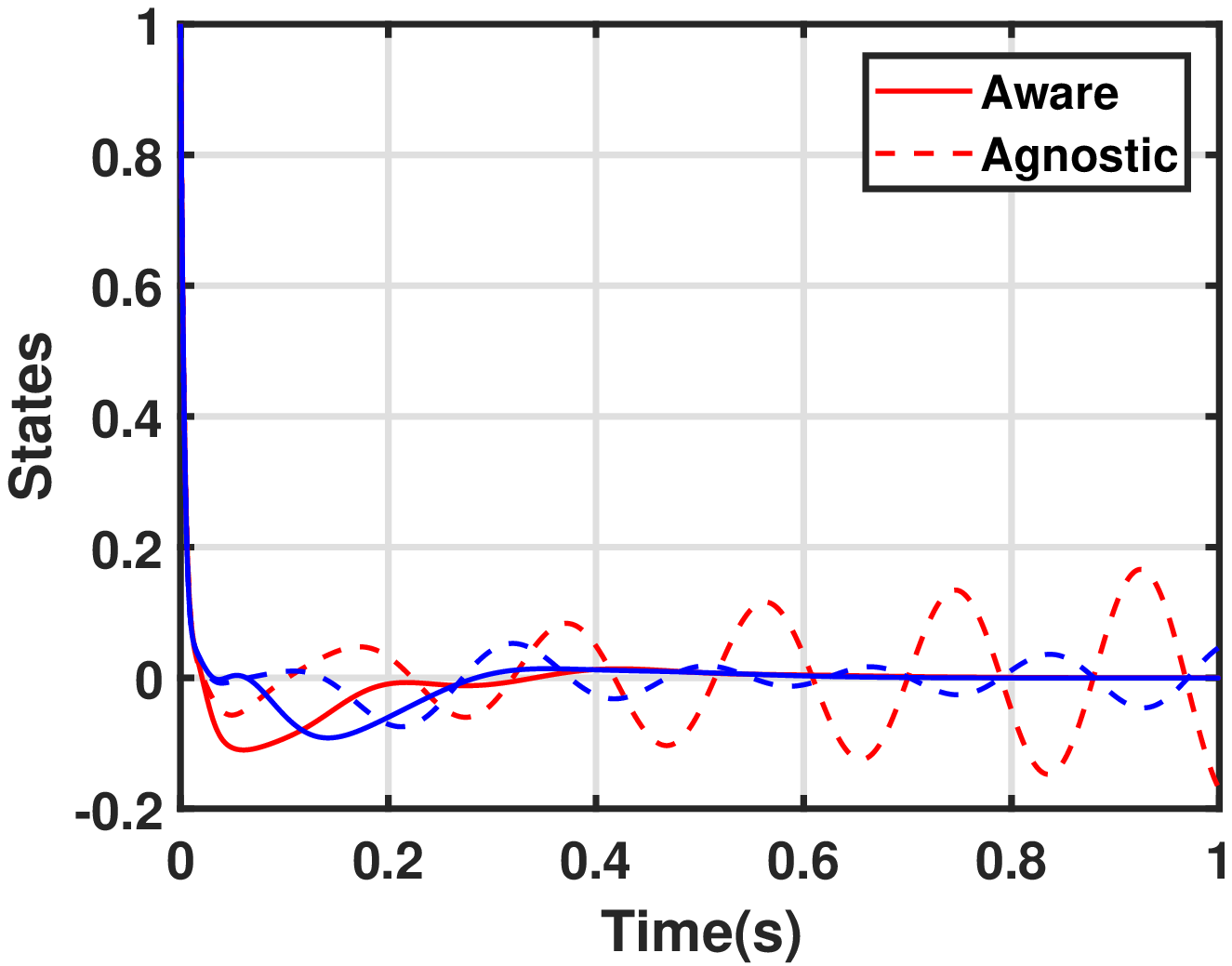}\\
{(a)}
\end{minipage}\\
 \begin{minipage}{0.5\linewidth}
\centering
\includegraphics[scale=0.4]{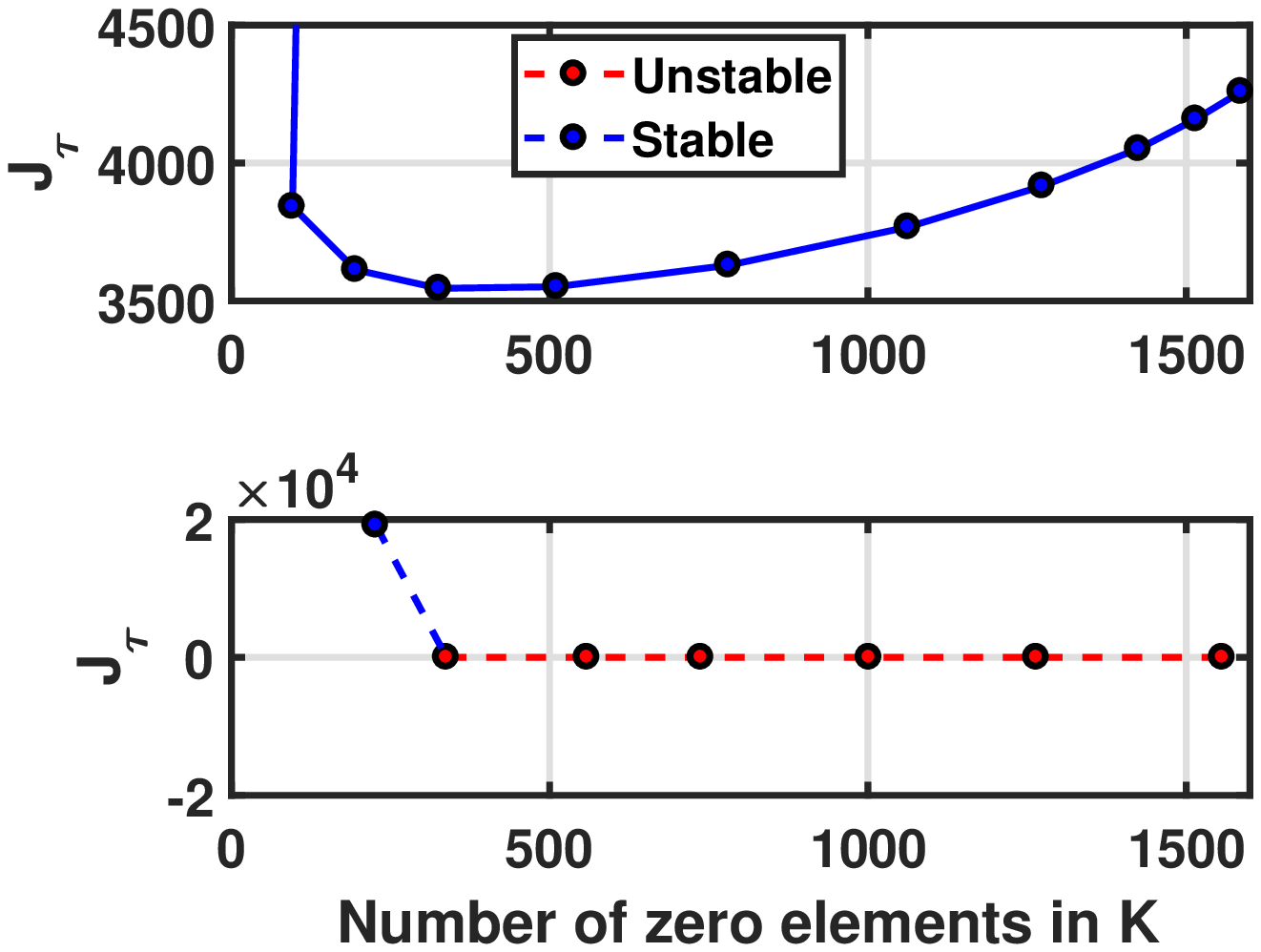}\\
 (b) Delay-aware (upper),\\ \hspace{0.85cm} Delay-agnostic (lower)
 \caption{Case-1 : \textbf{S1}}
 \label{fig1case1}
\end{minipage}\begin{minipage}{0.5\linewidth}
\centering
\includegraphics[scale=0.4]{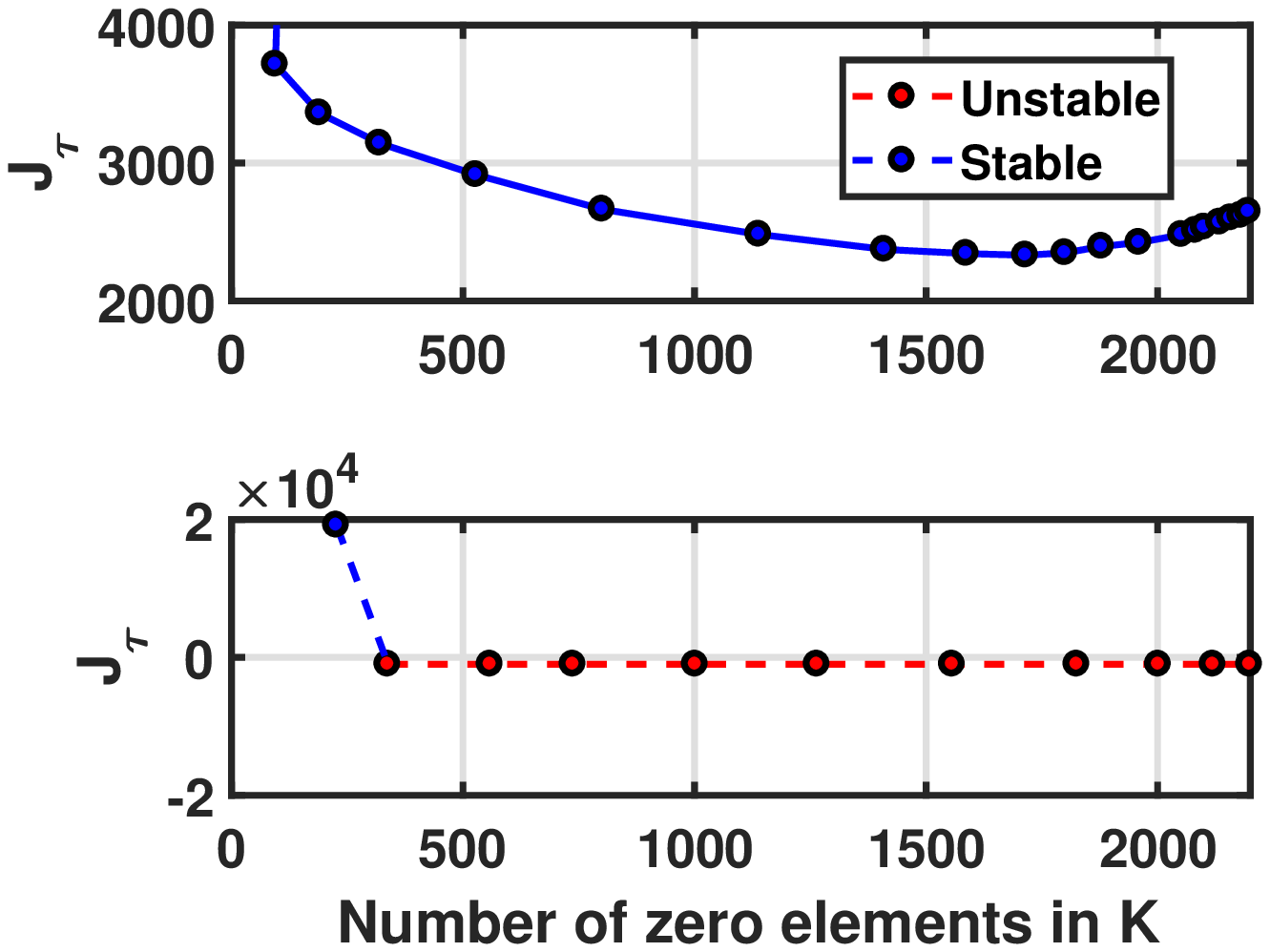}\\
 (b) Delay-aware (upper),\\ \hspace{0.85cm} Delay-agnostic (lower)
 \caption{Case-1 : \textbf{S2}}
 \label{fig2case1}
\end{minipage}
\end{figure}
\subsection{Case-1 : Element-wise Sparsity}
We consider a $50$-th order LTI system with a randomly generated state matrix $A$. The numerical value of $A$ can be found in the data set (Git Hub) in \citet{dataset1}. We assume $B=B_w=I_{50}$, $\tau_p=9.83$ ms, $c=956$ and $\kappa=0.01$. Algorithm \ref{algo1} starts with a full matrix $K_o$ and corresponding $\tau^m =19.7$ ms, and outputs a full $K'$ with the corresponding $\tau'=36$ ms such that $c_f=c$. Algorithm \ref{algo2} then sparsifies $K'$ while maintaining optimality. Fig. \ref{fig1case1}(a) and \ref{fig2case1}(a) compare the dynamic performance of the closed-loop states between the delay-aware and delay-agnostic algorithms for conditions \textbf{S1} (i.e., when $\tau$ is kept constant at $36$ ms), and  \textbf{S2} (i.e, when $\tau$ decreases from $36$ ms for $s=0$ to $12.42$ ms for $s=2252$\footnote{$s$ is the sparsity level of $K$ defined as $s=mn-\Card(K)$.}), respectively.  Fig. \ref{fig1case1}(b) and \ref{fig2case1}(b) show the comparison of $J_\tau$ versus $\Card(K)$ between the delay-agnostic and delay-aware algorithms for \textbf{S1} and \textbf{S2}, respectively. For \textbf{S1}, we can see in Fig. \ref{fig1case1}(b) that while $K$ obtained by both algorithms are of comparable sparsity levels, the one resulting from the delay-agnostic algorithm easily destabilizes the closed-loop system. In contrast, the delay-aware algorithm consistently maintains closed-loop stability and $\mc{H}_2$ optimality of the delayed system. Similar conclusions can be drawn for \textbf{S2} as seen in Fig. \ref{fig2case1}(b). Since $J_\tau$ is not defined for unstable systems, a nominal value of $J_\tau=-1000$ is allotted for representation.
\par The upper panel of Fig. \ref{fig1case1}(b) for \textbf{S1} shows an initial decrease in $J_\tau$ till it reaches a minimum, followed by a gradual increase as sparsity further increases. A similar trend is observed for \textbf{S2} in Fig. \ref{fig2case1}(b) but here the minimum value of $J_\tau$ is reached at a much higher sparsity level compared to \textbf{S1}. This is because \textbf{S1} forces $\tau$ to remain at the initial high value of $36$ ms, while in \textbf{S2} $\tau$ decreases with increasing sparsity, which aids in achieving a higher sparsity level for a lower $J_\tau$. The initial detrimental effect of increasing sparsity on $J_\tau$ is weaker compared to the aiding effect of the decreasing delay. After the minimum $J_\tau$ is reached, the effect of increasing sparsity takes over and $J_\tau$ starts increasing monotonically. Algorithm 2 successfully captures this minimum. 
 \subsection{Case-2 : Block Sparsity}
  Next, we consider a $10$-th order LTI system. The numerical values of the  randomly generated state matrix $A$ and the control input matrix $B=\texttt{BlkDiag}\left(B_1,B_2,B_3\right)$, with $B_1\in \mathbb{R}^{2\times 2}$, $B_2\in\mathbb{R}^{3\times 3}$ and $B_3\in \mathbb{R}^{5\times 5}$ are provided in \citet{dataset1}. We assume $B_w = I_{10}$. This case presents the extension of our design to a {\it block sparse} $K$. Again, as shown in Fig. \ref{fig1case3}(b) and Fig. \ref{fig2case3}(b) respectively, our algorithm guarantees closed-loop stability while the delay-agnostic algorithm fails to do so for both conditions \textbf{S1} and \textbf{S2}. Also, delay-aware algorithm for \textbf{S2} achieves higher sparsity for a lower $J_\tau$ compared to \textbf{S1} due to the flexibility in changing $\tau$.
  \begin{figure}[hbtp]
\begin{minipage}{0.5\linewidth}
\centering
\includegraphics[scale=0.4]{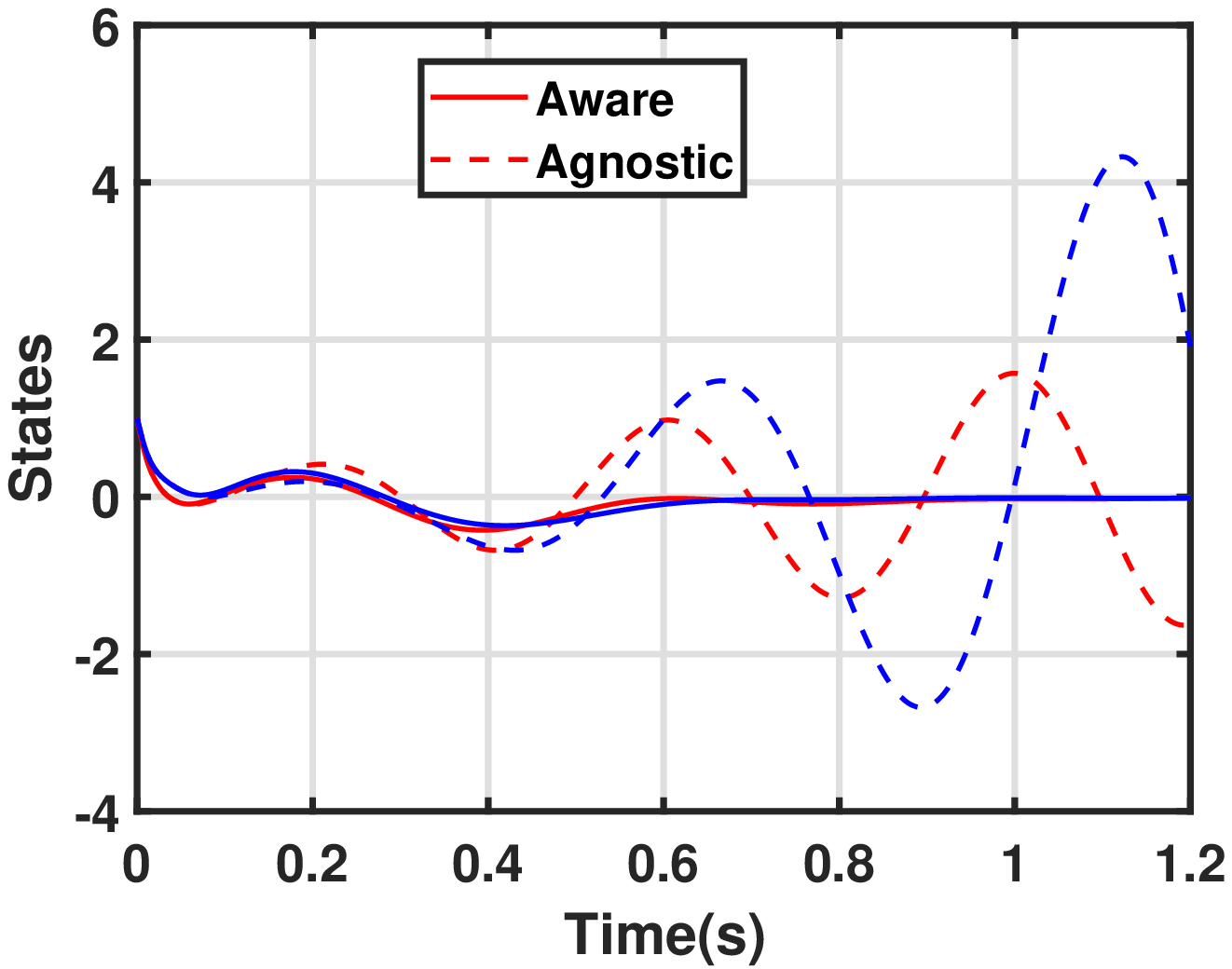}\\
(a)\\
\includegraphics[scale=0.42]{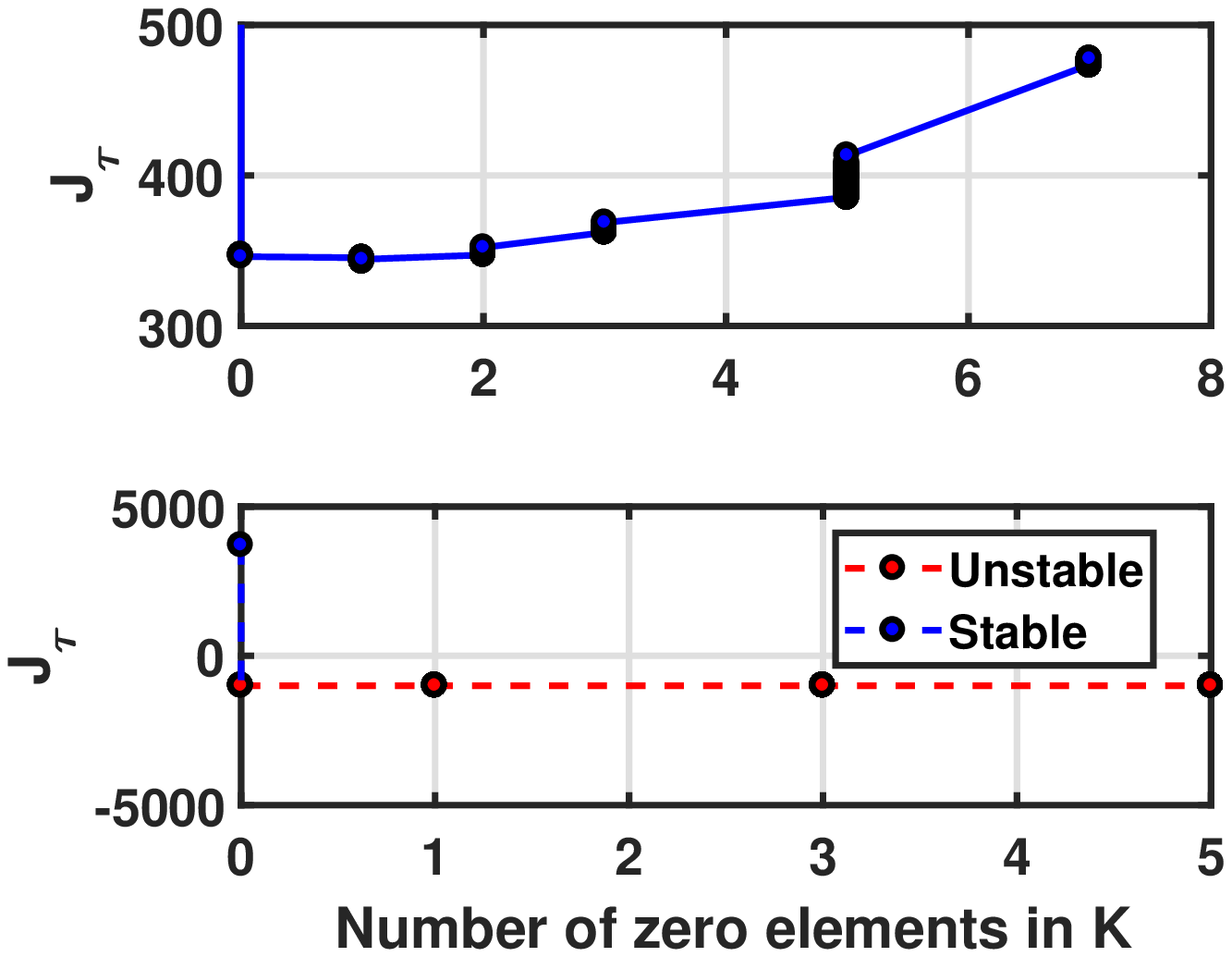}\\
(b) Delay-aware (upper), Delay-agnostic (lower)
\caption{Case-2 : \textbf{S1}}
\label{fig1case3}
\end{minipage}%
\begin{minipage}{0.5\linewidth}
\centering
\includegraphics[scale=0.4]{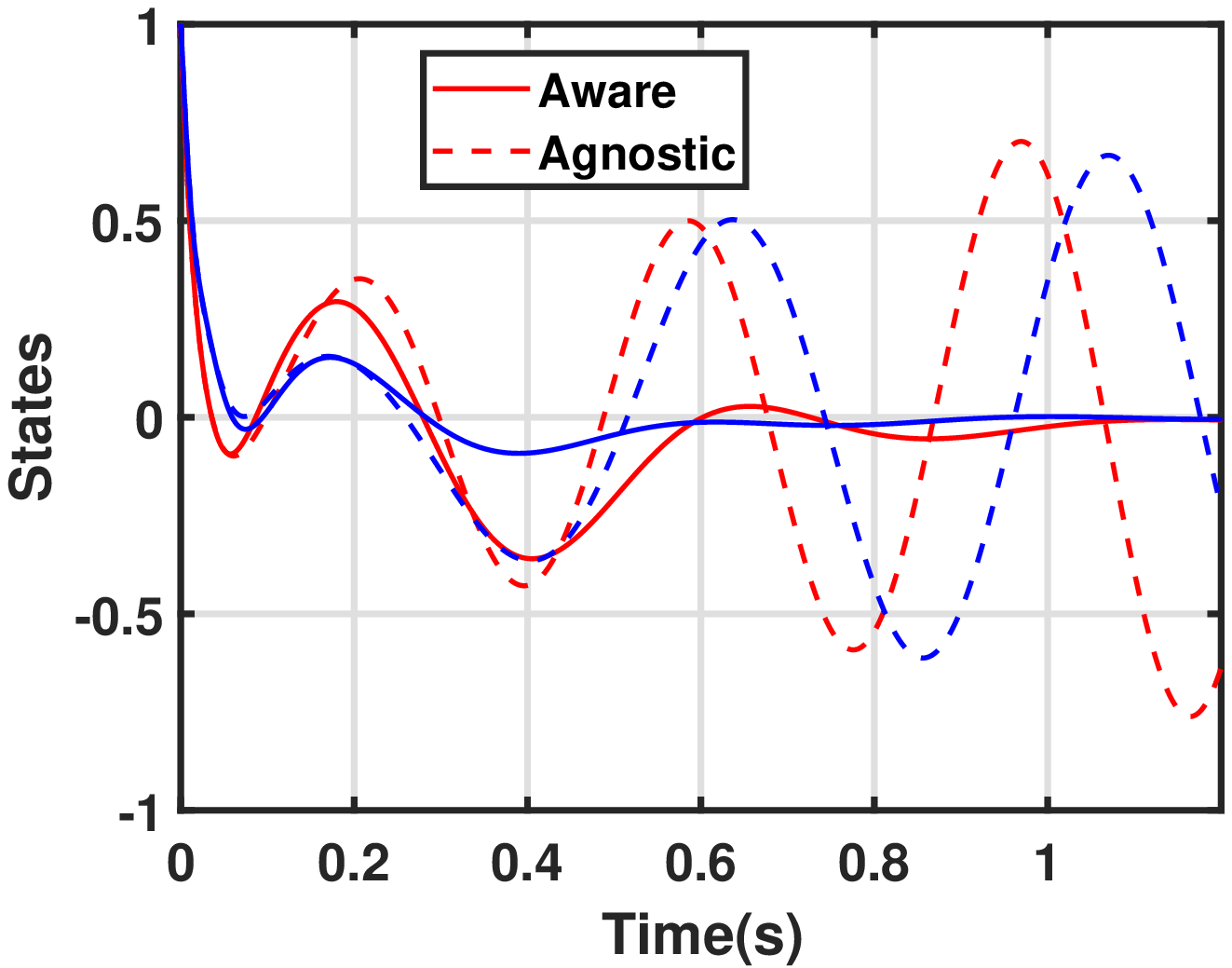}\\
(a)\\
\includegraphics[scale=0.4]{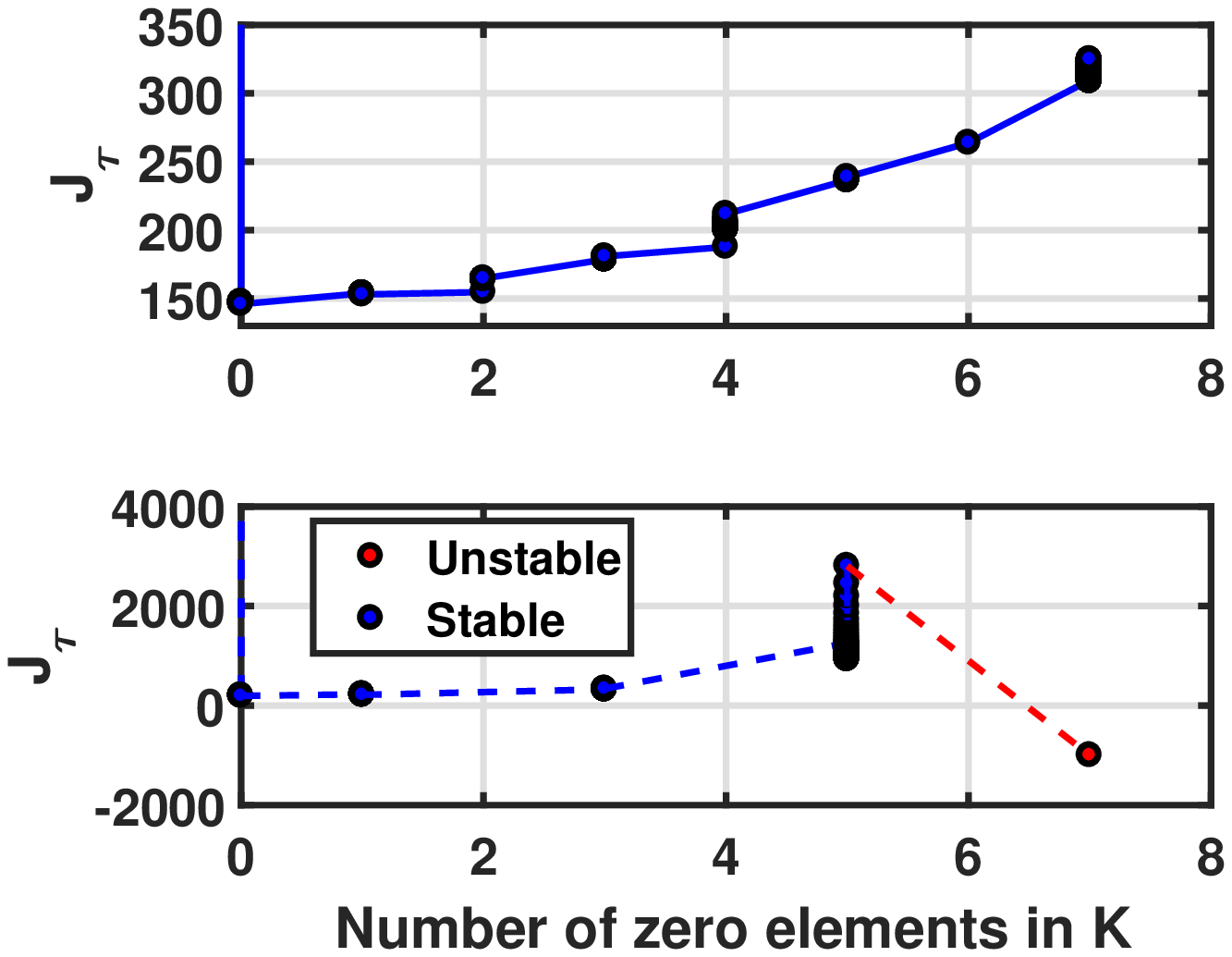}\\
(b) Delay-aware (upper), Delay-agnostic (lower)
\caption{Case-2 : \textbf{S2}}
\label{fig2case3}
\end{minipage}
\end{figure}
\subsection{Case-3 : \texorpdfstring{$K$}{K} is not stabilizing for the delay-free system}
One benefit of Algorithm \ref{algo1} is that we can obtain a stabilizing pair $(K',\tau')$ where $\tau'=\mc{Z}(\Card(K'),c,$ $\tau_p)$ lies beyond the delay margin interval $[0,\tau'^m]$ of $(A,B,K')$, as explained in Section \ref{sec:changingdelaywithsparsity}. This means that even for a relatively low value of the bandwidth $c$ and a high value of $\tau_p$, Algorithm \ref{algo1} can successfully find a stable pair $(K',\tau')$, thereby fulfilling the condition in \eqref{pairrequired}, and preventing the case where one may need to revise $c$ to a higher bandwidth $c'$. The trade-off is that such a starting pair for Algorithm \ref{algo2} may result in the \textbf{4-2} scenario of Section \ref{sec:changingdelaywithsparsity}, where we lose on $\mc{H}_2$ optimality. This case describes such a scenario. We consider a $10$-th order LTI system with randomly generated state matrix $A$ \citep{dataset1} and $B=B_w=I_{10}$. For $c=10.5$, $\tau_p=28.34$ ms, and $\kappa=0.01$ (compared to Case-1, $c$ is lower and $\tau_p$ is higher). Algorithm \ref{algo1} starts with a full matrix $K_o$ and the corresponding $\tau^m_o=56.5$ ms. It converges to the desired $\tau'=123.5$ ms corresponding to a full matrix $K'$. However, the pair $(K',\tau')$ is such that $K'$ is not stabilizing for $\tau=0$. Therefore, a comparison with the delay-agnostic algorithm is not possible. The sparsity behavior is seen to be similar to Case-1, as shown in Fig. \ref{fig1case2}(b) and \ref{fig2case2}(b) for \textbf{S1} and \textbf{S2}, respectively. For \textbf{S1}, $\tau$ is kept constant at $123.5$ ms, and the minimum $J_\tau=375.5$ is achieved at $s=12$. Thereafter, $J_\tau$ increases steeply with sparsity. In \textbf{S2}, $\tau$ decreases from $123.5$ ms to $97.87$ ms as we reach $s=58$. We see a steady decline in $J_\tau$ such that minimum $J_\tau=273.4$ is reached at $s=49$, thereby confirming that for \textbf{S2}, a lower $\mc{H}_2$ norm is obtained at a much higher sparsity level compared to \textbf{S1}. 
\begin{figure}[hbtp]
\centering
\begin{minipage}{0.5\linewidth}
\centering
\includegraphics[scale=0.4]{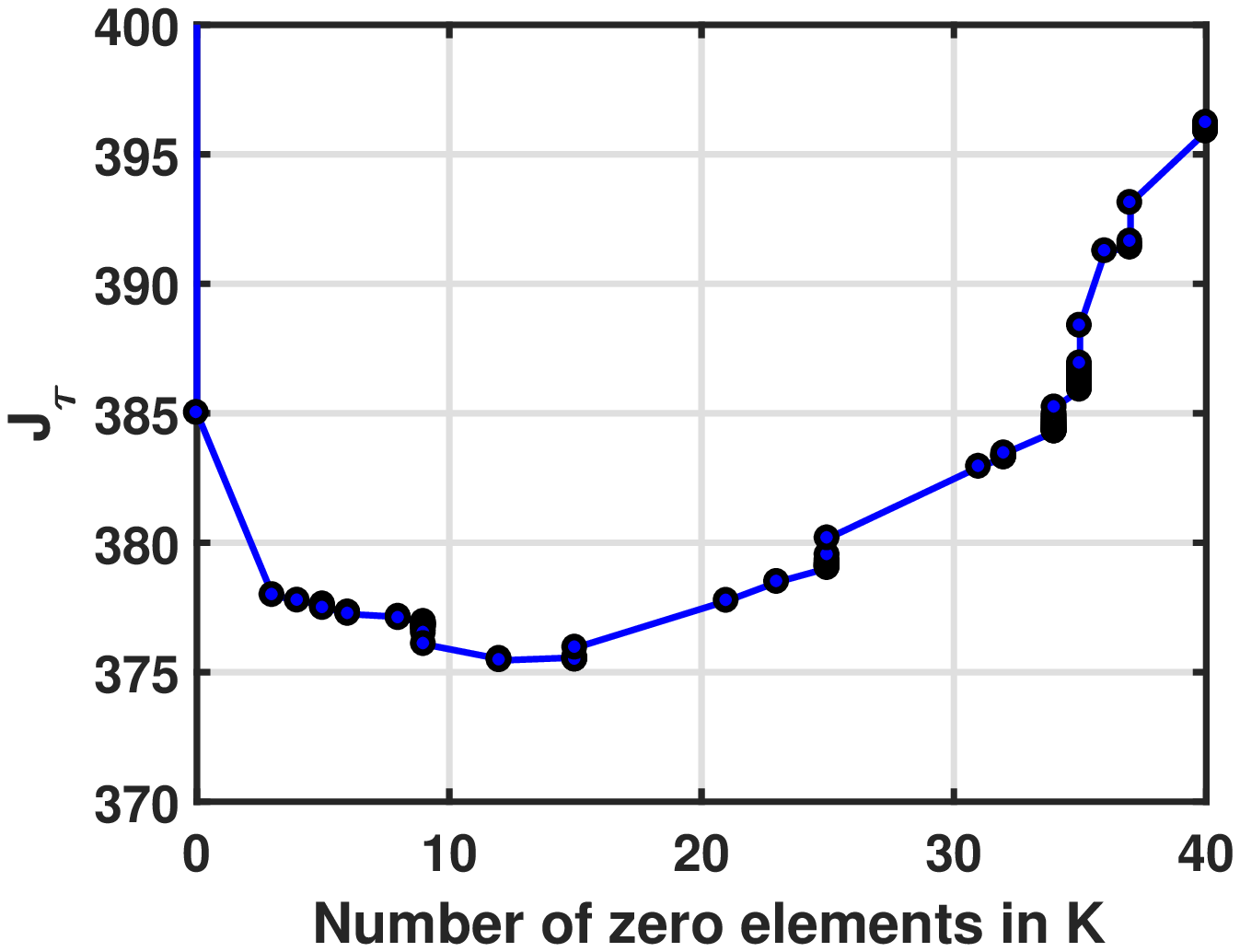}\\
\caption{Case-3 : \textbf{S1}}
\label{fig1case2}
\end{minipage}%
\begin{minipage}{0.5\linewidth}
\centering
\includegraphics[scale=0.42]{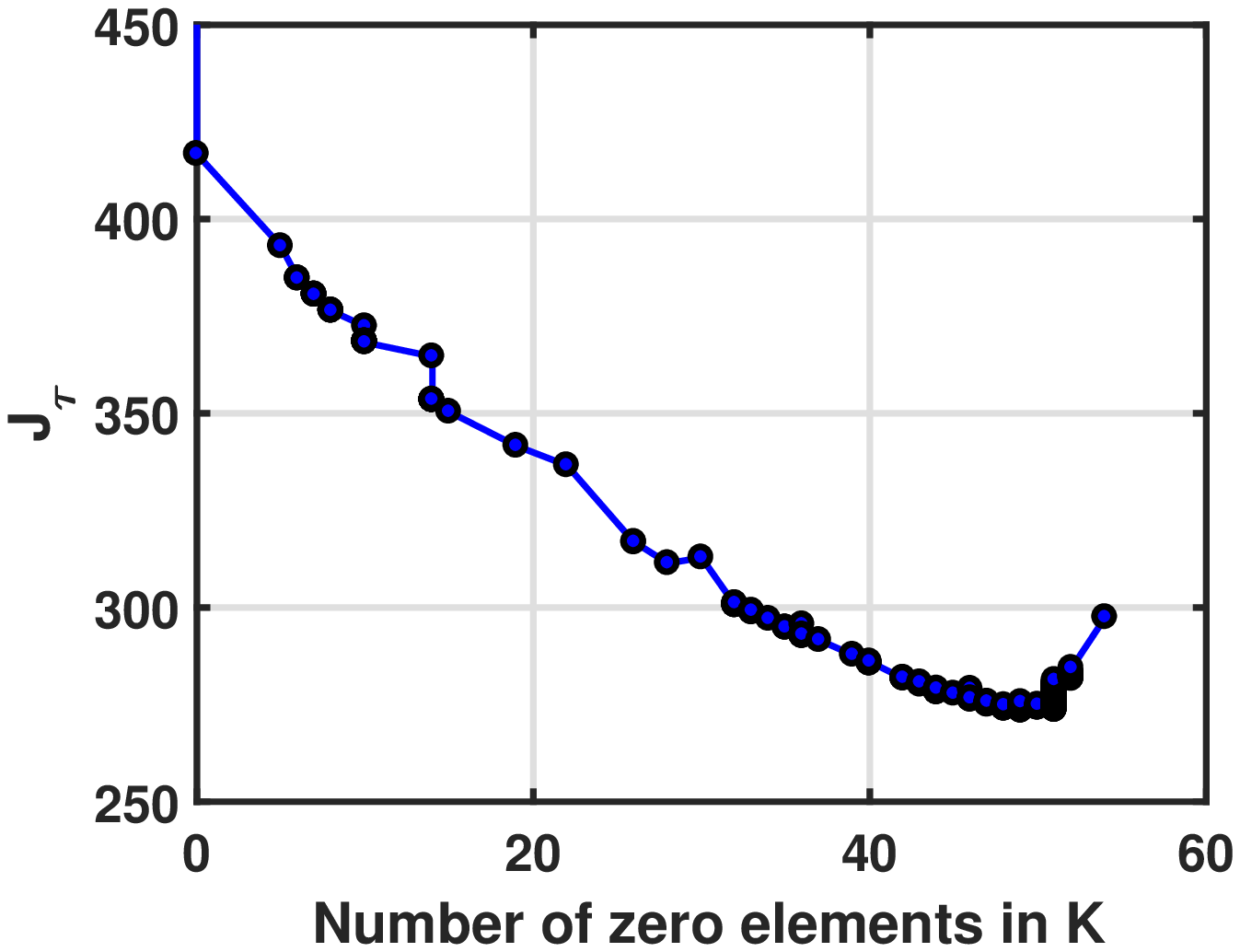}\\
\caption{Case-3 : \textbf{S2}}
\label{fig2case2}
\end{minipage}
\end{figure}
\section{Cooperative Coupled Sparsification}
\label{sec:coop}
Algorithms \ref{algo1} and \ref{algo2} have provided the designs for sparse $K$ when there is a single user consuming the entire available bandwidth. We next extend this design to a multi-user system, which consists of multiple dynamically decoupled LTI systems that share the same set of communication links for their individual closed-loop control. Thus, the closed-loop performance of the users are indirectly dependent on each other even though they are physically decoupled. This section presents a characterization of the {equivalent} $\mc{H}_2$ performance between the users followed by an algorithm that enables sparse controller designs by fairly allocating the number of available communication links.

\subsection{Problem Formulation}
\label{subsec:jointnetworkperformance}
Let there be ${N}$ users sharing $\ell$ communication links. Let the $i$-th user's closed-loop state-space model be represented as
\begin{align}
&\dot{x}_i(t) = A_i x_i(t) + B_i u_i(t) + B_{w_i} w_i(t), \ i\in\mathbb{N}_N,\\
&u_i(t)=K_i x_i(t-\tau_i),
\end{align}
where $\tau_i$ is the feedback delay associated with the $i$-th system, $x_i\in\mathbb{R}^{n_i}$ are the states, $u_i\in\mathbb{R}^{m_i}$ are the control inputs and $K_i\in \mathbb{R}^{m_i\times n_i}$ are the feedback gains. Ideally, one would need $\sum_{i=1}^N m_i n_i$ number of links to accommodate an optimal state-feedback controller for every user. If $\ell$ is less than this number then one must promote at least $\mathfrak{s}$ number of zero {entries} in $K_1,\ldots,K_N$, where $\mathfrak{s}=\sum_{i=1}^N m_in_i-\ell$.
Let $s_i\in\mathbb{Z}_+$ be the {sparsity level} of the $i$-th user, defined as the number of zero entries promoted in $K_i$. This means that $\sum_{i=1}^N s_i = \mathfrak{s}$. $K_i$ can be designed for different sparsity levels $s_i$ using Algorithm \ref{algo1} and \ref{algo2}. Since multiple structures of $K_i$ are possible for each such $s_i$, we choose that $K_i$ which results in the lowest $J_{\tau_i}$. This generates a $J_{\tau_i}$ vs $s_i$ characteristic curve for every user $i$. An example of this characteristic curve is seen in Fig. \ref{fig1case1}(b). To normalize these curves over $N$ users, let $J^o_{i}=\te{min} \ J_{\tau_i}$ be called as the \textit{nominal performance} of the $i$-th user. We define the \textit{performance cost ratio} $r_i$ of the $i$-th user as
\begin{equation}
\label{R}
r_i(s_i) = \frac{J_{\tau_i}(s_i)}{J^o_{i}}.
\end{equation}
The ratio curve $r_i (s_i)\geq 1 $ is the normalized characteristic of the $i$-th user; see Fig. \ref{fig:CPS2}(a). Let $s=[s_1,\ldots,s_N]^T\in\mathbb{R}^N$ be the vector of sparsity levels that will be allotted to the users, with the corresponding performance ratios being $r(s)=[r_1(s_1),\ldots,r_N(s_N)]^T\in\mathbb{R}^N$. We denote $H(r_1(s_1),\ldots,$ $r_N(s_N))$ as the \textit{variance} of the $\mc{H}_2$ performance ratio between the users defined as
\begin{align}
H=\frac{1}{N} \sum_{i=1}^N (r_i(s_i)- \bar{r}(s))^2 =\underbrace{\frac{r^T(s) r(s)}{N}}_{f(r(s))} - \underbrace{\frac{(1^T {r}(s))^2}{N^2}}_{g(r(s))}. \label{H}
\end{align} 
The variance $H\geq 0$ measures the fairness of link allotment. We state our problem as:
\begin{itemize}
    \item \textit{Find $s=[s_1,\ldots,s_N]$ to minimize the variance $H(r(s))$ of $\mc{H}_2$ performance ratio between the users such that $\mathfrak{s}=\sum_{i=1}^{N} s_i$, where $\mathfrak{s}$ is a known number.}
\end{itemize}
To present the above objective as an optimization problem, we define the following sets:
\begin{align}
\mc{S}_i := \ \{s_i \in \mathbb{Z}_+ : \exists \ K_i&\in\mc{V}_{\tau_i}, \ s_i=m_in_i - \Card(K_i)\},\\
 \mc{S}:=& \ \mc{S}_1\times \mc{S}_2 \times \cdots \times \mc{S}_N, \label{originalS} \\
 \mc{S}_M := & \ \mc{S}_1 +\mc{S}_2 + \cdots + \mc{S}_N.
\end{align}
$\mc{S}_i$ is the set of possible sparsity levels of the $i$-th user. $\mc{S}$ is the feasible set of our problem. The optimization problem can be stated as:
\begin{subequations}
\label{optimnew}
\begin{align}
\boldsymbol{\mc{O}}_4 : \ &\underset{s}{\te{min}} \ F=\ f(r(s)) - g(r(s)) + \sigma h(r(s)) \label{optim1}\\
&\hspace{0.5cm}\te{s.t.} \hspace{2cm} 1^T s \leq \mathfrak{s}, \label{optim2} \\
&\hspace{3.5cm} \ s\in\mc{S}, \label{optim3} 
\end{align}
\end{subequations}
where $\sigma$ is a positive scalar parameter, $\mathfrak{s}\in\mc{S}_M$ and $h(r(s))= \sum_{i=1}^N r_i(s_i)$. The term $h(\cdot)$ is added to ensure that the sum of the ratios $r_i$ do not become too large. For instance, between $r=[1.33,1.33]^T$ and $r=[1.32,1.31]^T$ for the same $\mathfrak{s}$, the latter would be preferable even though the former has a lower $H$. We choose $0<\sigma \ll 1$ to ensure that $\bs{\mc{O}}_4$ places a higher relative importance on minimizing the variance.
\begin{figure}
    \centering
    \includegraphics[scale=0.5]{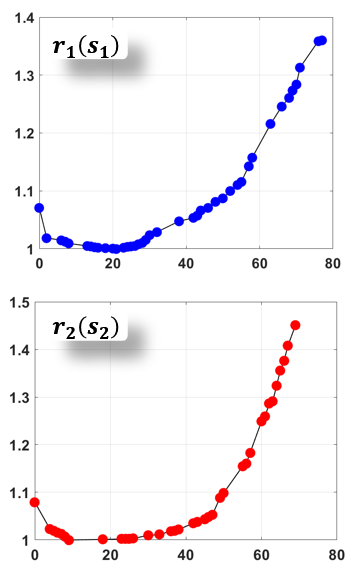}  \hspace{2cm}  \includegraphics[scale=0.5]{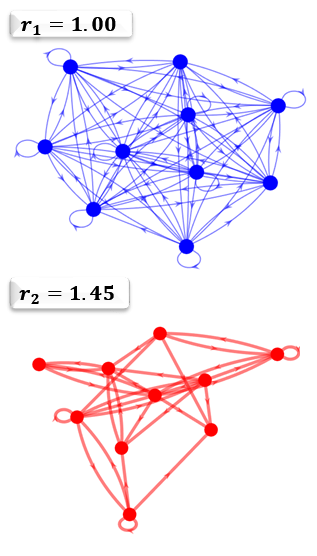}\hspace{2cm}\includegraphics[scale=0.5]{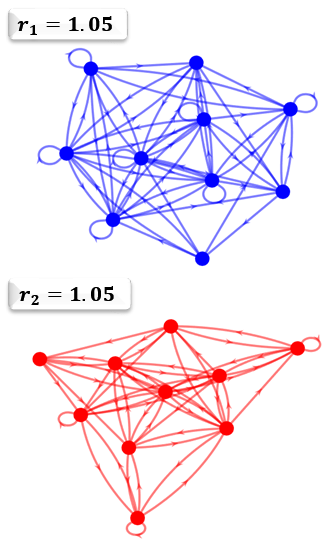}\\
    (a) \hspace{4cm} (b) \hspace{4cm} (c)
    \caption{This example shows cooperative network sparsification for N=2 users. The network has a total of $111$ available links. Each user can have a maximum of $\ell=100$ links. Fig (a) shows the $r_i(s_i)$ characteristic of the two users. Fig (b) shows the communication graphs of the two users (blue for user 1, red for user 2) at the start of Algorithm 3, when user 1 has $80$ links while user 2 has $31$ links, but their performance ratios $r_1=1$ and $r_2=1.45$ are highly disproportionate. Fig (c) shows the respective communication graphs at the end of Algorithm 3, where now user 1 has $53$ links and user 2 has $58$ links, smoothing out the disparity in their performance ratios, both of which are now $1.05$.}
    \label{fig:CPS2}
\end{figure}
\par The optimization problem in \eqref{optimnew} is defined over $\mc{S}\subset\mathbb{Z}^N$, and is therefore, combinatorial. Moreover, our objective in \eqref{optim1} is defined as a functional on the desired domain $\mc{S}$, in contrast to a standard optimization problem. We first propose a framework for relaxing the domain of $\bs{\mc{O}}_4$ from $\mathbb{Z}^N$ to $\mathbb{R}^{N}$, followed by introducing a convex-concave programming (CCP) based algorithm that solves the relaxed program.
\subsection{Relaxing the Optimization Constraints}
Since $\mc{S}_i$ is a discrete set $\forall$ $i\in\mathbb{N}_N$, $r_i(s_i)$ can be represented as a piece-wise affine function with discontinuities at every point in $\mc{S}_i$. Without loss of generality, we can assume that the discrete set $\mc{S}_i$ for the $i$-th user has $p_i$ number of elements ordered as
\begin{equation}
    s_{i,1} < s_{i,2} < \cdots < s_{i,p_i}. \label{sij}
\end{equation}
We define the continuous intervals $\mc{R}_{i,j}=[s_{i,j}, s_{i,j+1}]$, $j\in\mathbb{N}_{p_i-1}$ such that $\mc{R}_i :=  \bigcup^{p_i-1}_{j=1} \mc{R}_{i,j}$ for the $i$-th user.
\par\noindent\textbf{Example (a)} - If user $1$ has $p_1=4$ sparsity levels $\mc{S}_1=\{0,3,4,5\}$, then $\mc{R}_{1,1}=[0,3]$, $\mc{R}_{1,2}=[3,4]$, $\mc{R}_{1,2}=[4,5]$, and $\mc{R}_1=[0,5]$. If user $2$ has $p_2=4$ with $s_2=\{0,2,4,5\}$ and $\mathfrak{s}=8$, then the feasible set of $\bs{\mc{O}}_4$ is $\{(4,4),(3,5)\}$.
\par\noindent We can represent $r_i(s_i)$ for the $i$-th user as follows:
\begin{equation}
r_i(s_i) = a_{i,j} s_i + b_{i,j} , \ \forall \ s_i \in \mc{R}_{i,j}, \ j\in\mathbb{N}_{p_i-1}.\label{piecewise}
\end{equation}
Since the piecewise-affine characteristics of $r_i(s_i)$ is defined over finite real set $\mc{R}_i$ $\forall$ $i$, $\bs{\mc{O}}_4$ will search over closed and bounded sets. We next reformulate the integer program of $\bs{\mc{O}}_4$ as an MILP with the non-convex objective of \eqref{optim1}. Table \ref{tab:variables} defines the optimization variables for an $i$-th user with the sparsity level $s_i\in\mc{R}_i$ and the corresponding $r_i(s_i)$ piece-wise defined over $p_i-1$ intervals. The $p_i-2$ binary variables $\delta_{i,j} \in \{0,1\}$ determine the interval $\mc{R}_{i,j}$ for $s_i$. $\tilde{z}_i\in\mathbb{R}$ represents the corresponding piece of $r_i(s_i)$. Using Table \ref{tab:variables}, we define vector $v_i$ for every $i\in\mathbb{N}_N$, which represents the $i$-th user's decision variable in the reformulated MILP:
\begin{equation}
 v_i=[\tilde{z}_i ,z_{i,1},\ldots,z_{i,p_i-2},s_i,\delta_{i,1},\ldots,\delta_{i,p_i-2}]. 
\end{equation}
\begin{table}[hbtp]
    \centering
    \begin{tabular}{|p{0.2\linewidth}|p{0.68\linewidth}|}
 \hline
      $\delta_{i,j}\in \{0,1\}$,  & $\delta_{i,j} =1 \iff s_i \geq s_{i,j+1}$ \\
      $j\in\mathbb{N}_{p_i-2}$ &  \\
      \hline
      $z_{i,1} \in \mathbb{R}$ & $= \begin{cases}
a_{i,2} s_i + b_{i,2}, \ \ \te{if} \ \delta_{i,1}=1\\
a_{i,1} s_i + b_{i,1}, \ \ \te{else}
\end{cases}$\\
$z_{i,j} \in \mathbb{R}$, & $ = \begin{cases}
(a_{i,j+1} - a_{i,j})s_i + (b_{i,j+1} - b_{i,j}), \ \ \te{if} \ \delta_{i,j}=1\\
0, \ $\text{otherwise}$.
\end{cases}$\\
$j\in\mathbb{N}_{p_i-2}\backslash 1$ & \\
\hline 
$\tilde{z}_i\in\mathbb{R}$ &  $=\sum_{j=1}^{pi-2} z_{i,j}$\\
\hline
$d^{max}_i\in\mathbb{Z}_+$ & $=s_{i,p_i} -s_{i,1}$\\
$d^{min}_i\in\mathbb{Z}_+$ & $=s_{i,2}-s_{i,p_i}$ \\
\hline
 $t_{i,j}\in\mathbb{R}$, $j=1,2$ & $\underset{s_i\in\mc{R}_{i}}{\te{max}} \ \ \ \ \  a_{i,j} s_i - b_{i,j}$\\
  $t_{i,j}\in\mathbb{R}$, $j=3,\ldots,p_i-1$ & $\underset{s_i\in\mc{R}_{i}}{\te{max}} \ (a_{i,j}-a_{i,j-1}) s_i - (b_{i,j}-b_{i,j-1})$\\
\hline
    \end{tabular}
\caption{Optimization variables for MILP reformulation of $\bs{\mc{O}}_4$}
\label{tab:variables}
\end{table}
 \begin{table}[hbtp]
    \centering
\begin{tabular}{|p{0.255\linewidth}|p{0.63\linewidth}|}
\hline
  $j=1,\ldots,p_i-1$   & $d^{max}_i \delta_{i,j} - s_i \leq -s_{i,j+1} +d^{max}_i$\\[-0.2cm]
& $(d^{min}_i-\epsilon)\delta_{i,j} +s_i \leq  s_{i,j+1}-\epsilon$  \\
\hline
 $j=2,\ldots,p_i-1$ & $ \delta_{i,j}-\delta_{i,l} \leq 0 $, $l<j$  \\
 \hline
 $j=3,\ldots,p_i$ &  $ \mp t_{i,j} \delta_{i,j-1} \pm z_{i,j-1} \leq 0,$\\[-0.2cm]
  &$\mp t_{i,j} \delta_{i,j-1} \pm z_{i,j-1} \mp (a_{i,j}-a_{i,j-1}) s_i \mp (b_{i,j} \mp b_{i,j-1}) \leq -t_{i,j}$ \\
  \hline
$j=3,\ldots,p_i$ & $\pm(t_{i,2}-t_{i,1})\delta_{i,1}\mp z_{i,1}\pm a_{i,2} s_i \leq \pm ( t_{i,2}\- - t_{i,1})\mp b_{i,2}$\\
\hline
$j=3,\ldots,p_i$ & $\pm(t_{i,2}-t_{i,1})\delta_{i,1}\mp z_{i,1}\pm a_{i,1} s_i \leq \mp b_{i,1}$\\
\hline
 $j=3,\ldots,p_i$ & $\pm \tilde{z}_i\mp \sum_{j=1}^{\bar{p}_i} z_{i,j} \leq 0$ \\
\hline
\end{tabular}
\caption{Reformulating the constraints of $\bs{\mc{O}}_4$}
\label{tab:milpequations}
\end{table}
Table \ref{tab:milpequations} translates the definitions of Table \ref{tab:variables} into mixed integer linear inequalities, thereby, relaxing the combinatorial constraints of $\bs{\mc{O}}_4$. The inequalities in Table \ref{tab:milpequations} can be written in the form of compact LMIs $G_i v_i \leq g_i$, $i\in\mathbb{N}_N$. We reformulate $\bs{\mc{O}}_4$ as
\begin{subequations}
\label{optimnew1}
\begin{align}
\bs{\mc{O}}_5 : \ &\underset{s}{\te{min}} \ F=\ f(r(s)) - g(r(s)) + \sigma h(r(s)) \label{optim1new} \\
&\text{s.t.}\hspace{1cm} G_i v_i(s_i) \leq g_i, \ \ i\in\mathbb{N}_N,\\
&\hspace{1.7cm} 1^T s \leq \mathfrak{s}. \label{sconstraint} 
\end{align}
\end{subequations}
Let $\mathbb{S}_\mathfrak{s}$ define the solution set of $\bs{\mc{O}}_5$ for $\mathfrak{s}\in\mc{S}_M$ as:
\begin{equation}
\mathbb{S}_\mathfrak{s}=\left\{(s_1,\ldots,s_N)\in\mathbb{R}^N: \sum_{i=1}^N s_i = \mathfrak{s}, s_i\in \mc{R}_i \right\}.\label{setS}
\end{equation}
Since $\mathbb{S}_\mathfrak{s}$ is the Cartesian product of closed and bounded sets $\mc{R}_i$, $i\in\mathbb{N}_N$, it is also closed, bounded and non-empty. The objective function $F(r_1(s_1),\ldots,r_N(s_N))$ is continuous over the sets $\mc{R}_i$, which are, in turn, continuous over $s_i$ from \eqref{piecewise}. Therefore, $F$ is continuous over $s_i$, $i\in\mathbb{N}_N$ and a global minima $s^*=(s^*_1,\ldots,s^*_N)$ of $F$ with $\sum_{i=1}^N s^{*}_i=\mathfrak{s}$ exists in $\mathbb{S}_\mathfrak{s}$ from Weierstrass's theorem.
\par\noindent While $\bs{\mc{O}}_4$ is defined over $\mc{S}\subset\mathbb{Z}^N_+$, $\bs{\mc{O}}_5$ is defined over $\mathbb{S}\subset\mathbb{R}^N_+$, where
\begin{equation}
    \mathbb{S}:=\bigcup\limits_{\mathfrak{s}\in\mc{S}_M} \mathbb{S}_\mathfrak{s} = \mc{R}_1 \times \cdots \times \mc{R}_N.
\end{equation}
This implies that $\mc{S} \subset \mathbb{S}$. We next modify the $r_i(s_i)$ characteristics of all users to ensure that the local minima of $\bs{\mc{O}}_5$ always lie in $\mc{S}$ despite being defined over the much larger set $\mathbb{S}$.
\par\noindent \textbf{Example (a) [Contd.]} The local minima for $\bs{\mc{O}}_5$ should belong to $\{(4,4),(3,5)\}$. However, since $\mc{R}_1=\mc{R}_2=[0,5]$ are continuous real sets on which the piece-wise affine characteristics of \eqref{piecewise} are defined, the local minimum can very well be at some point $(4.5, 3.5)$, which is unacceptable.
\subsection{Modifying \texorpdfstring{$r_i(s_i)$}{ri(si)} Characteristics to Accommodate Integer Sparsity Levels }
To ensure that any local minima of $\bs{\mc{O}}_5$ belongs to $\mc{S}$ instead of $\mathbb{S} \backslash \mc{S}$, the following theorem modifies the piece-wise affine $r_i(s_i)$ at $s_i\notin \mc{S}_i$ $\forall$ $i\in\mathbb{N}_N$.
\begin{theorem}
\label{theorem5}
Let $r_i(s_i)$ be defined on intervals $\mc{R}_{i,j}=[s_{i,j},s_{i,j+1}]$, $j\in\mathbb{N}_{p_i-1}$, as given in \eqref{sij} and \eqref{piecewise}. With $\epsilon>0$, we partition each interval $\mc{R}_{i,j}$ for the $i$-th user as
\begin{align}
\mc{R}^{(1)}_{i,j} =&[s_{i,j}, s_{i,j}+\epsilon), \ \mc{R}^{(2)}_{i,j} =[s_{i,j}+\epsilon,s_{i,j+1} -\epsilon]\label{r1}\\
\mc{R}^{(3)}_{i,j} =&(s_{i,j+1}-\epsilon,s_{i,j+1}], \label{r3}
\end{align}
such that $\mc{R}_{i,j}=\bigcup_{q=1}^{3} R^{(q)}_{i,j}$. Let $D\in\mathbb{R}$ be a scalar parameter. We modify $r_i(s_i)$ to $\tilde{r}_i(s_i)$ as
\begin{align}
\tilde{r}_i(s_i) = \begin{cases}
r_i(s_{i}), \  &\te{If} \ s_i=s_{i,j}, \\
 i D, \ &\te{If} \  s_i=s_{i,j}\pm\epsilon,
\end{cases}\label{tildeR}
\end{align}
which is represented by the piecewise affine function
\begin{equation}
    \tilde{r}_i(s_i) = \tilde{a}^{(q)}_{i,j} s_i + \tilde{b}^{(q)}_{i,j},\ s_i\in \mc{R}^{(q)}_{i,j}.
\end{equation}
Denoting $r^{(m)} = \te{max} \{r^{(m)}_i\}$, $r^{(m)}_i =\te{max}_{\mc{S}_i}\{r_i(s_i)\}$, if $D \geq \frac{2 N^3 (r^{(m))^2}}{\epsilon}$ and $0<\epsilon < \text{min}(\sigma,\nicefrac{1}{N})$, then all the local minima for $\bs{\mc{O}}_5$ are contained inside the set $\mc{S}$.
\end{theorem}
\begin{proof}
The proof is stated in Appendix (Section \ref{subsec:proofoftheorem5}). 
\end{proof}
\par\noindent The piecewise affine function $\tilde{r}_i(s_i)$ retain $r_i$ values for all $s_i\in\mc{S}_i$ following \eqref{tildeR}. The corresponding modified intervals for the $i$-th user $\mc{R}^{(k)}_{i,j}$, $k\in\mathbb{N}_3$, $j\in\mathbb{N}_{p_i-1}$ are used to determine the optimization variables in Table \ref{tab:variables}, followed by the linear inequalities in Table \ref{tab:milpequations} in the form of the modified LMIs $\tilde{G}_i \tilde{v}_i \leq \tilde{g}_i$. 
\par\noindent We next use convex-concave procedure (CCP) to relax the non-convex objective function in \eqref{optim1new} \citep{ccp}. This is done by considering any point $s_{(0)}=[s_{1_0},\ldots,s_{N_0}]^T \in \mc{S}_{(\mathfrak{s})}$, where $\mc{S}_{(\mathfrak{s})}\subset \mc{S}$ is defined as
\begin{equation}
    \mc{S}_{(\mathfrak{s})} := \{s\in\mc{S} : \sum_{i=1}^N s_i = \mathfrak{s}\},
\end{equation}
and bounding the concave part of the DC objective $-g(s)$ by its convex relaxation around this point as
\begin{align}
\label{ccp}
\hat{g}(s;s_{(0)})=g(s_{(0)}) + \nabla g(s_0)^T (s-s_{(0)}).
\end{align}
In that case, the relaxed objective function for $\bs{\mc{O}}_5$ can be written as
\begin{align}
    \label{Fhat}
    \Hat{F}(s;s_{(0)}) = f(s) - \Hat{g}(s;s_0) + \sigma h(s).
\end{align}
\begin{algorithm}[hbtp]
 \caption{Cooperative Sparsification}
 \label{algo4}
 \begin{algorithmic}[1]
\State \textbf{Input:} Number of links present, $\ell$.
\State \hspace{1.1cm} Number of entries in $K_1,\ldots,K_N$ to be
\Statex \hspace{1.1cm} removed, $\mathfrak{s}=\sum^{N}_{i=1}m_in_i-\ell$.
\State \textbf{Input:} Initial $s_0=[s_{1_0},\ldots,s_{N_0}]^T$ and $r_0(s_0)$.
\State \textbf{Input:} LMI matrices $\tilde{G}_i$, $\tilde{g}_i$ $\forall$ $i\in\mathbb{N}_N$.
\State \textbf{Set:} $s^*=s_0$, $F^*=F(s_0)$.
\For{Fixed number of Steps}
\State Construct $s=[\varrho_1 \tilde{v}_1,\ldots, \varrho_N \tilde{v}_N]$, where 
\Statex $\varrho_i\in\mathbb{R}^{2p_i-2}$ has $1$ as its $p_i$-th entry, $0$ otherwise. 
\State \textbf{Solve the following program :}
\State $\underset{\tilde{v}_1,\ldots,\tilde{v}_N}{\te{argmin}} \  f(s) - \hat{g}(s;s_0)  + \sigma h(s)$
\State s.t. \hspace{0.5cm} $\tilde{G}_i \tilde{v}_i \leq \tilde{g}_i$, $i\in\mathbb{N}_N$,
\State \hspace{1cm} $\sum_i^N \varrho_i \tilde{v}_i=\mathfrak{s}$. 
\State \textbf{Check : } $F(s)<F^*$. If yes, set $F^*=F(s)$, 
\Statex \hspace{2cm}$s^*=s$.
\State \textbf{Set : } $s_0=s$.
\EndFor
\State \textbf{Result : $s^*$ }.
\end{algorithmic}
 \end{algorithm}
\par\noindent Algorithm \ref{algo4} shows the steps for solving $\bs{\mc{O}}_5$. The following proposition provides the conditions under which Algorithm \ref{algo4} converges.
\begin{proposition}
\label{proposition1}
For a given $\mathfrak{s}\in\mc{S}_M$, Algorithm \ref{algo4} iteratively solves $\bs{\mc{O}}_5$ such that
\par (a) the difference between $F(s_{(j)})$ and its convex upper bound $\hat{F}(s_{(j)};s_{(j-1)})$ with respect to the known $s_{(j-1)}$ is given as:
\begin{equation}
\delta_{j,j-1}=\frac{1}{N^2} \left(\sum_{i=1}^N \big(r_i(s_{i_j})-r_i(s_{i_{j-1}})\big) \right)^2,
\end{equation}
where $r_i(s_{i_j})$ is the performance ratio of the $i$-th user for the $j$-th iteration.
\par (b) Beginning from an initial guess $s_{(0)}\in\mc{S}_{(\mathfrak{s})}$, Algorithm \ref{algo4} will
converge to a local minima $s_{(k)}\in\mc{S}_{(\mathfrak{s})}$ such that
\begin{equation}
    0 \leq  F(s_{(k)}) \leq F(s_{(0)}) - \sum_{j=1}^k \delta_{j,j-1}, \label{propb}
\end{equation}
\end{proposition}
\begin{proof}
The proof is given in Appendix (Section \ref{subsec:proofofprop1}).
\end{proof}

\section{Examples}
\label{sec:simulations2}
To test the efficiency of Algorithm 3, simulations are carried out for $N=2$ and $N=3$ users. For each case, we assign a ranking system to each $s\in\mc{S}$, where $\mc{S}$ is the solution set of our original problem $\bs{\mc{O}}_4$. The reason for this is the reformulation of $\bs{\mc{O}}_4$ to $\bs{\mc{O}}_5$ due to which all the discrete points in $\mc{S}$ now become the local minima for Algorithm \ref{algo4}. Ranking of these local minima is needed to evaluate the performance of the algorithm. Separate ranking is created for each individual $\mathfrak{s}\in\mc{S}_M$. Increasing ranks are allotted according to the increasing $F(r(s))$. 
\par For $N=2$ case, the state matrices for both users are randomly chosen as $A_i\in\mathbb{R}^{10\times 10}$ with $B_i=B_{w_i}=I_{10}$, $i=1,2$. For $N=3$ case, the three systems have $n_1=5$, $n_2=4$ and $n_3=5$ number of states. The state matrices for these systems are randomly generated $A_i\in\mathbb{R}^{n_i\times n_i}$, with $B_i=B_{w_i} = I_{n_i}$. Fig. \ref{figgg1} and \ref{figgg2} show the simulation results for $N=2$ and $N=3$, respectively, where the x-axis represents the given $\mathfrak{s}$, and, the y-axis shows the absolute difference between the values of $F(r(s))$ for rank-1 and rank-2 solutions, represented by $\Delta F = |F_{\te{rank 1}}- F_{\te{rank 2}}|$. The color of the scatter plot represents the rank of the solution obtained from Algorithm \ref{algo4} with `high' referring to rank-1 (high rank) and `low' referring to the worst-ranked solution.
\begin{figure}[hbtp]
\begin{minipage}{0.5\linewidth}
\centering
\includegraphics[scale=0.5]{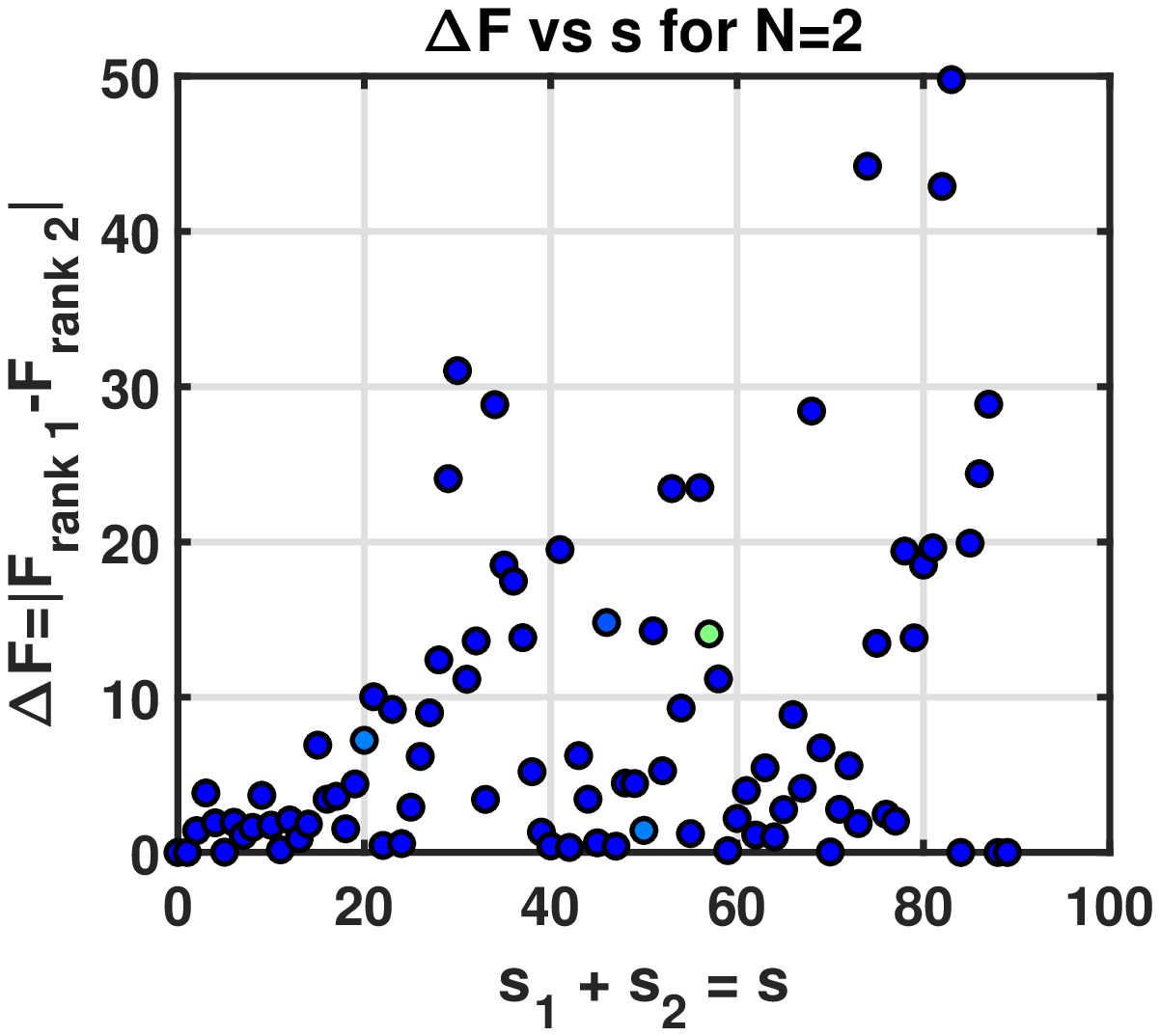}
\caption{N=2}
\label{figgg1}
\end{minipage}\begin{minipage}{0.5\linewidth}
\centering
\includegraphics[scale=0.5]{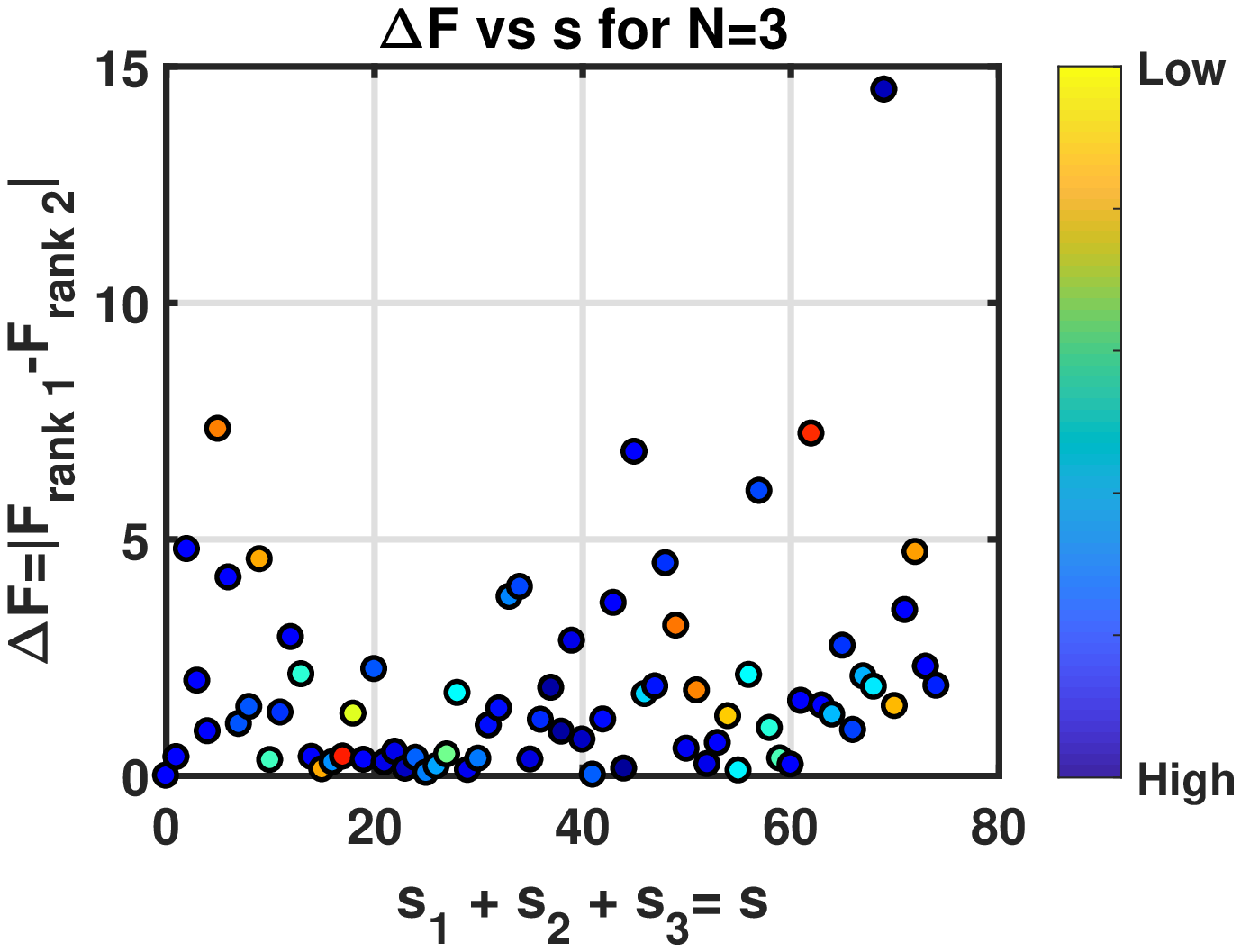}
\caption{N=3}
\label{figgg2}
\end{minipage}
\end{figure}
\par For $N=2$ case, while user-1 achieves $p_1=41$ sparsity levels $0<2<\cdots<76<77$, user-2 achieves $p_2=36$ sparsity levels ranging from $0<4<\cdots<67 <69$. Fig. \ref{figgg1} shows that Algorithm \ref{algo4} achieves high ranking solutions for most of the $\mathfrak{s}$ values. However, low ranking solutions can be obtained when $\Delta F=|F_{rank 1} - F_{rank 2}|$ is too small due to the steps of the branch and bound method that is used to solve the MILPs. Similar trend of low ranking solutions are observed for $N=3$ as well, as shown in Fig. \ref{figgg2}. 
\par A three-dimensional plot of the solution set for $\mathfrak{s}=89$ in the case of $N=2$ users is shown in Fig. \ref{fig:3dplot}, where the blue dots denote the solution set of $\bs{\mc{O}}_4$, while the connecting blue curve creates the solution set for $\bs{\mc{O}}_5$. The global minima is represented by the red dot. The $r_i(s_i)$ graphs for each user are shown in Fig. \ref{fig:CPS2}. In summary, this example shows the case where a total of $111$ links are divided, beginning from an initial allocation of $80$ and $31$ links to user-1 and 2 respectively, which corresponds to a high discrepancy between their performance ratios $r_1=1.00$ and $r_2=1.45$. Starting from this initial guess, Algorithm \ref{algo4} reaches the global minima in the second iteration, which decreases the number of links allotted to user-1 from $80$ to $53$, and, increases the number of links of user-2 from $31$ to $58$, resulting in $r_1\approx r_2=1.05$, which ensures the lowest possible disparity in their performances.
\begin{figure}[hbtp]
    \centering
    \includegraphics[scale=0.35]{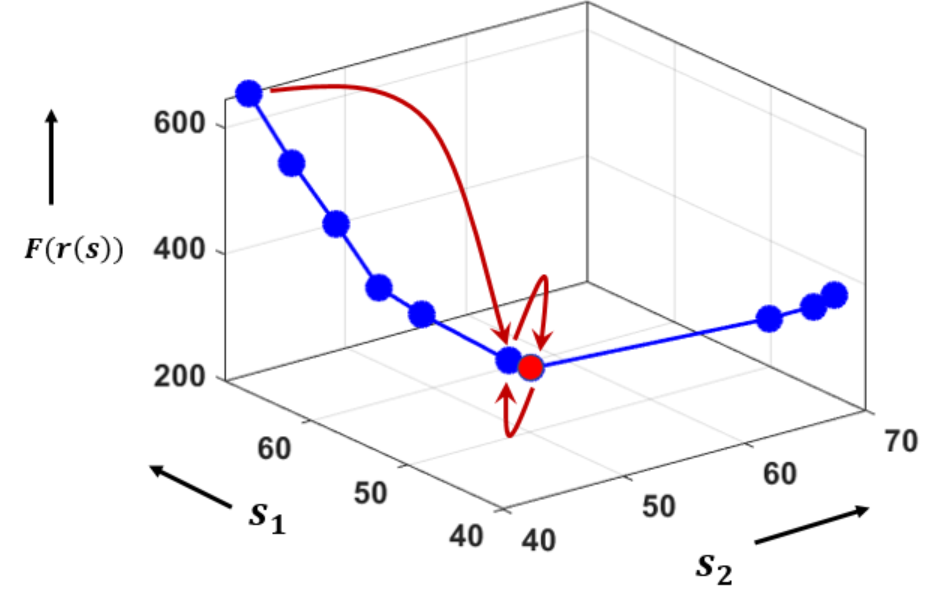}
    \caption{3-D plot of the solution set of $\bs{\mc{O}}_4$ for $\mathfrak{s}=89$ when $N=2$. The blue dots show the local minima, while the red dot shows the global minimum. The progress of Algorithm \ref{algo4} is shown by the red arrows.}
    \label{fig:3dplot}
\end{figure}
\section{Conclusion}
We develop algorithms to design sparse $\mc{H}_2$ optimal controllers for an LTI system in the presence of feedback delay. Our algorithms bring out the interplay between delays, sparsity and $\mc{H}_2$ objective. Both sparsity and $\mc{H}_2$ norm are facilitated when we include the flexibility in controlling the delay as compared to when the delay is kept constant. While traditional sparsity promoting algorithms destabilize the delayed system in most cases, ours preserves stability all throughout. In the second half of the paper, we extend these algorithms to a multiple user scenario, and assign a fair number of communication links to each user with the aim of reducing the variance of their $\mc{H}_2$ performances. Numerical simulations are presented to showcase the efficiency of all algorithms. Our future research will include extension of these strategies to a CPS with multiple delays, and analyzing their sensitivity to variations in network traffic using ideas from sparse stochastic optimal control.
\section{Appendix}
\label{sec:appendix1}
\subsection{Proof of Theorem \ref{theorem2}}
\label{subsec:proofoftheorem1}
\begin{proof}
Substituting $K=K_o + \Delta K$, $\tau = \tau_o + \Delta \tau$, $P = P_o + \Delta P$, $\cdots$, $S_{1}=S_{1o} + \Delta S_1$ in \eqref{lmi1} and assuming $\Delta \mc{F}: = \texttt{blkdiag}(B\Delta K,B\Delta K,$ $B\Delta K,B\Delta K)$, the NMI $\Phi_{\te{N}}\succ 0$ can be written as
\begin{align}
&\Phi_\te{N} =\ \Phi_0 + \underbrace{\Phi_1}_{{\te{bilinear}}} + \underbrace{\Phi_{2}}_{{\te{bilinear}}}+ \underbrace{\Phi_{3}}_{{\te{bilinear}}}+ \underbrace{\Phi_4}_{{\te{trilinear}}} \succ 0, \label{cbmi1original}
\end{align}
where $\Phi_0$ is given as
\begin{subequations}
\begin{align}
&\Phi_0 = \left[\begin{array}{c:c:c:c}
\Delta_0 & \Phi_{0_1} & \Phi_{0_2} & \Phi_{0_3} \\
\hdashline
\star & S_1 & \Phi_{0_4} & \Phi_{0_5} \\
\hdashline
\star & \star & \Phi_{0_6} & 0\\
\hdashline
\star & \star & \star & 3(S_0-S_1)
\end{array}\right],\\
&\Phi_{0_1} = Q_1 - P_o BK - \Delta P B K_o\nn\\
&\Phi_{0_2}=-\frac{\tau}{2} \big[A^T (Q_{0o} + Q_{1o}) + (R_{00o} + R_{01o}) -\frac{{\tau}_o}{2} \big[A^T ({\Delta Q_{0}} + {\Delta Q_{1}}) + ({\Delta R_{00}} + {\Delta R_{01}})\big]+(Q_0 - Q_1)\nn \\
&\Phi_{0_3}=\frac{\tau}{2} \big[A^T (Q_{0o} - Q_{1o}) + (R_{00o} - R_{01o})\big]+\frac{{\tau}_o}{2} \big[A^T ({\Delta Q_{0}} -{\Delta Q_{1}}) + ({\Delta R_{00}} - {\Delta R_{01}})\big] \nn\\
&\Phi_{0_4}=-\frac{\tau}{2} \big[K^T_o B^T (Q_{0o} + Q_{1o}) + (R^T_{01o} + R_{11o})-\frac{{\tau}_o}{2} \big[K^T_o B^T ({\Delta Q_{0}} +{\Delta Q_{1}}) + \Delta K^T B^T ({Q_{0o}} +{Q_{1o}})   \nn\\
& \hspace{7cm}+ ({\Delta R^T_{01}}+ {\Delta R_{11}})\big]\nn\\
&\Phi_{0_5}=\frac{\tau}{2} \big[K^T_o B^T (Q_{0o} - Q_{1o}) - (R^T_{01o} - R_{11o})+\frac{{\tau}_o}{2} \big[K^T_o B^T ({\Delta Q_{0}} -{\Delta Q_{1}}) + \Delta K^T B^T ({Q_{0o}} -{Q_{1o}}) \nn\\
& \hspace{7cm} - ({\Delta R^T_{01}} - {\Delta R_{11}})\big]\nn\\
&\Phi_{0_6}=\tau (R_{00o} - R_{11o}) + (S_0 - S_1) + {\tau}_o({\Delta R_{00}} - {\Delta R_{11}}),\nn
\end{align}
\end{subequations}
and $\Phi_i$, $i=1,2,3,4$ are given as
\begin{subequations}
\begin{align}
&\Phi_1= {\mc{B}_1}{\Delta \mc{F}} + {\Delta \mc{F}^T}{\mc{B}_1^T}, \ \Phi_2=\frac{{\Delta\bar{\tau}}}{2} ({\mc{B}_2} + {\mc{B}^T_2}),\label{phi12}\\
&\Phi_3={\Delta {\tau}}({\mc{B}_3} + {\mc{B}^T_3}), \ \Phi_4={\Delta {\tau}} \big({\mc{B}_4}{\Delta \mc{F}} + {\Delta\mc{F}^T}{\mc{B}_4^T}\big),\label{phi34}
\end{align}
\end{subequations}
with $\mc{B}_i$, $i=1,2,3,4$ defined as follows:
\begin{subequations}
\begin{align}
&\mc{B}_1=\left[\begin{array}{c:c:c:c}
\bb{0} & -{\Delta P} & \bb{0} &\bb{0}\\ \hdashline
\bb{0} & \bb{0} & \bb{0} & \bb{0}\\ \hdashline
\bb{0} & -\frac{{\tau}_o}{2} ({\Delta Q_0} + {\Delta Q_1})^T - \frac{\Delta {\tau}}{2}({Q_{0o}}+{Q_{1o}})^T & \bb{0} & \bb{0} \\ \hdashline
\bb{0} &\hspace{0.3cm} \frac{{\tau}_o}{2} ({\Delta Q_0} - {\Delta Q_1})^T + \frac{\Delta {\tau}}{2}({Q_{0o}}-{Q_{1o}})^T & \bb{0} & \bb{0}
\end{array}\right]\\
&\mc{B}_2=\left[\begin{array}{c:c:c:c}
\bb{0} & \bb{0} & - A^T ({\Delta Q_{0}} + {\Delta Q_1}) &  A^T ({\Delta Q_{0}} - {\Delta Q_1})  \\\hdashline
\bs{0} & \bb{0} & -K^T_o B^T({\Delta Q_{0}} +{\Delta Q_{1}})  &  K^T_o B^T ({\Delta Q_{0}} -{\Delta Q_{1}})\\\hdashline
\bb{0} & \bb{0} & ({\Delta R_{00}} - {\Delta R_{11}}) & \bb{0}\\ \hdashline
\bb{0} & \bb{0} & \bb{0} & \bb{0}
\end{array}\right]\\
&\mc{B}_3=\left[\begin{array}{c:c:c:c}
\bb{0} & \bb{0} & -\frac{1}{2} ({\Delta R_{00}} + {\Delta R_{01}}) & \frac{1}{2} ({\Delta R_{00}} - {\Delta R_{11}})  \\\hdashline
\bb{0} & \bb{0} & \frac{1}{2}({\Delta R^T_{01}} +{\Delta R_{11}})  &- \frac{1}{2}({\Delta R^T_{01}} -{\Delta R_{11}})\\\hdashline
\bb{0} & \bb{0} & \bb{0} & \bb{0}\\ \hdashline
\bb{0} & \bb{0} & \bb{0} & \bb{0}
\end{array}\right]\\
&\mc{B}_4=\left[\begin{array}{c:c:c:c}
\bb{0} & \bb{0} & \bb{0} & \bb{0} \\\hdashline
\bb{0} & \bb{0} & \bb{0} & \bb{0}\\ \hdashline
\bb{0} & -\frac{1}{2}({\Delta Q_0} + {\Delta Q_1})^T & \bb{0} & \bb{0}\\ \hdashline
\bb{0} & \hspace{0.4cm}\frac{1}{2}({\Delta Q_0} - {\Delta Q_1})^T & \bb{0} & \bb{0}
\end{array}\right].
\end{align}
\end{subequations}
The matrices $\Phi_1$, $\Phi_3$ and $\Phi_4$ are symmetric Hamiltonian \citep[Theorem 4.3.1]{nandini}, while $\Phi_2$ is symmetric, but not Hamiltonian. Therefore, the following relations can be derived using \citep[Lemma 1, Theorem 2]{nandini}:
\begin{subequations}
\label{phibound}
\begin{align}
&|\lambda_{min}(\Phi_1)|=\|\Phi_1\| = \|(\mc{B}_1) (\Delta\mc{F})\| \leq \|{\Delta {K}}\|\|{\mc{B}_1}\| \label{phi1bound} \\
 & |\lambda_{min}(\Phi_2)| \leq \|\Delta{\tau} \mc{B}_2\| \leq |{\Delta{\tau}}| \|{\mc{B}_2}\|, \label{phi2bound} \\
&|\lambda_{min}(\Phi_3)|=\|\Phi_3\| = \|(\Delta {\tau})(\mc{B}_3)\| \leq |{\Delta {\tau}}|\|{\mc{B}_3}\|\label{phi3bound}\\
&|\lambda_{min}(\Phi_4)|=\|\Phi_4\| = \|\Delta{\tau}\mc{B}_4\Delta \mc{F}\| \leq |{\Delta{\tau}}|\|{\Delta {K}}\|\|{\mc{B}_4}\|.\label{phi4bound}
\end{align}
\end{subequations}
If the NMI $\Phi_\te{N}=  \sum_{i=0}^{4} \Phi_i \succ 0$, then $\sum_{i=1}^4 \lambda_{\te{min}}\left( \Phi_i \right)  \succ 0$ as well \citep[Theorem 1.2]{weyl}. Therefore, \eqref{phi1bound}-\eqref{phi4bound} can be used to derive the relaxed convex constraint for ensuring $\Phi_\te{N}\succ 0$. However, the bounds in \eqref{phi1bound}-\eqref{phi4bound} are still bilinear. Hence, two additional convex constraints are introduced next as:
\begin{align}
|\Delta {\tau}| \leq \gamma, \  \|\Delta K \| \leq \eta,  \label{deltataubound}
\end{align}
where $\eta>0$ and $\gamma>0$ are scalar design variables. Substituting \eqref{deltataubound} in \eqref{phi1bound}-\eqref{phi4bound}, one can obtain \eqref{constraint1} to \eqref{constraint4} as the final relaxed convex constraint set for $\Phi_\te{N} \succ 0$. 
\par Next, $\Psi_N\succ 0$ is relaxed. Assuming $0<\gamma \ll 1$, which implies $\Delta {\tau} \ll 1$, we can write:
\begin{align}
\frac{1}{{\tau}} = \frac{1}{{\tau}_o + \Delta {\tau}} = \frac{{\tau}_o - \Delta {\tau}}{{\tau}^2_o - (\Delta {\tau})^2}\approx \frac{1}{{\tau}_o} - \frac{\Delta {\tau}}{{\tau}^2_o}. \label{tauconditioning}
\end{align}
Using \eqref{tauconditioning}, $\Psi_\te{N}\succ 0$ can be rewritten as:
\begin{subequations}
\begin{align}
 &\hspace{2cm}   \Psi_\te{N}=\Psi_0+\Psi_1 \succ 0,\\
&\Psi_0 = \left[\begin{array}{c:c:c}
P & Q_0 & Q_1 \\ \hdashline
\star & R_{00} + \frac{1}{{\tau}_o}{S_0} - \frac{{\Delta{\tau}}}{{\tau}_o^2}S_{0o} & R_{01}\\ \hdashline
\star & \star & R_{11}+\frac{1}{{\tau}_o}{S_1} -\frac{{\Delta{\tau}}}{{\tau}_o^2}S_{1o}
\end{array}\right],\label{psi0}\\
&\Psi_1 =  {\Delta {\tau}} \mc{C}_1, \ \mc{C}_1= \left[\begin{array}{c:c:c}
\bb{0} & \bb{0}& \bb{0}\\ \hdashline
\star & -\frac{1}{\bar{\tau}_o^2} {\Delta S_o} & \bb{0}\\ \hdashline
\star & \star & -\frac{1}{\bar{\tau}_o^2}{\Delta S_1}
\end{array}\right].\label{psi1}
\end{align}
\end{subequations}
Using \eqref{deltataubound}, \eqref{psi0}-\eqref{psi1} and the same procedure as in \eqref{phibound} and \eqref{deltataubound}, $\Psi_{\te{N}}\succ 0$ can be equivalently written as \eqref{constraint5}-\eqref{constraint6} which imply $\Phi_{\te{N}} \succeq \epsilon I \implies \Phi_\te{N} \succ 0$ and $\Psi_\te{N} \succeq \epsilon I \implies \Psi_\te{N} \succ 0$.
\end{proof}
\subsection{Proof of Theorem \ref{theorem5}}
\label{subsec:proofoftheorem5}
\begin{proof}
Let the solution of $\bs{\mc{O}}_5$ be denoted as $s=(s_1,\ldots,s_N)\in \mathbb{S}_\mathfrak{s}$ for a given $\mathfrak{s}\in\mc{S}_M$, i.e., $\sum_{i=1}^N s_i =\mathfrak{s}$. Denoting $r(s) = [r_1(s_1),\ldots,r_N(s_N)]^T$, $F(r(s))=H(r(s)) + \sigma h(r(s))$ is denoted as $F(s) = H(s) + \sigma h(s)$. Since $\bs{\mc{O}}_5$ is solved using $\tilde{r}(s)$, the objective function $F(\tilde{r}(s))$ is denoted as $\tilde{F}(s)$.
\par We define the set $\mc{\tilde{S}} := \mc{\tilde{R}}_1 \times \cdots\times \mc{\tilde{R}}_N$, where $\mc{\tilde{R}}_i = \bigcup_{j=1}^{p_i-1}\bigcup_{q=1,3} R^{(q)}_{i,j}$ for the $i$-th user. We further define its subset $\mc{\tilde{S}}_{(\mathfrak{s})}$ $:= \{s\in\mc{\tilde{S}}:$ $\sum_{i=1}^N s_i = \mathfrak{s} \}$. Note that the desired solution set $\mc{S}_{(\mathfrak{s})}$ defined in \eqref{originalS} is such that $\mc{S}_{(\mathfrak{s})}\subset \mc{\tilde{S}}_{(\mathfrak{s})}\subset \mathbb{S}_\mathfrak{s}$. Since $\bs{\mc{O}}_5$ is evaluated over $\mathbb{S}_\mathfrak{s}$, we first prove that with the given bounds on $D$ and $\epsilon$, the local minima of $\tilde{F}(s)$ cannot belong to $\mathbb{S}_\mathfrak{s}/\mc{\tilde{S}}_{(\mathfrak{s})}$. We further prove that the local minima only belong to the subset $\mc{S}_{(\mathfrak{s})}$ of $\mc{\tilde{S}}_{(\mathfrak{s})}$. For any $s\in \mathbb{S}_\mathfrak{s}$, $\tilde{F}(s)$ is defined as: 
\begin{equation}
    \tilde{F}|_{s\in\mathbb{S}} = F(\tilde{r}_1(s_1),\ldots,\tilde{r}_{N}(s_N)), \label{66}
\end{equation}
where $s_i=s^*_i \pm \epsilon_i$ is an $\epsilon_i>0$ perturbation on the nearest sparsity level $s^*_i\in\mc{S}_i$ such that $\sum_{i=1}^N s^*_i \pm \epsilon_i=\mathfrak{s}$. The modified curve function of \eqref{tildeR} can be written as: 
\begin{equation}
\label{rfunc}
\tilde{r}_i(s_i) =   r_i (s^*_i) + (i D - r_i(s^*_i) ) \bar{\epsilon}_i,
\end{equation}
where $\bar{\epsilon}_i=({\epsilon_i}/{\epsilon}) \in [0,1]$ $\forall  \ i$. If $s\in\mathbb{S}_\mathfrak{s}/\mc{\tilde{S}}_{(\mathfrak{s})}$, then for $1\leq k \leq N$ users, $\bar{\epsilon}_i=1$; for the remaining $N-K$ users $\bar{\epsilon}_i \in [0,\epsilon]$. The objective function value for $s\in\mathbb{S}_\mathfrak{s}/\mc{\tilde{S}}_{(\mathfrak{s})}$ is evaluated as (dropping arguments from $\tilde{r}_i$ and $r_i$):
\begin{align}
\tilde{F}(s)|_{s\in\mathbb{S}_\mathfrak{s}/\mc{\tilde{S}}_{(\mathfrak{s})}} &= H(\tilde{r}_1,\ldots,\tilde{r}_N) + \sigma h(\tilde{r}_1,\ldots,,\tilde{r}_N )\label{HHtilde}
\end{align}
Using \eqref{rfunc}, the given bounds on $D$ and $\epsilon$, we obtain:
\begin{align}
 & \tilde{F}(s)|_{s\in\mathbb{S}_\mathfrak{s}/\mc{\tilde{S}}_{(\mathfrak{s})}} > \sigma (N + (D-r_1)) = \sigma N+ \frac{\sigma}{\epsilon}N^3 (r^{(m)})^2  > (r^{(m})^2 + N r^{(m)} > \underset{s^*\in\mc{S}_{(\mathfrak{s})}}{\text{max}} \ \tilde{F}(s^*),\label{solution}
\end{align}
 which shows that local minima cannot belong to $\mathbb{S}_\mathfrak{s} / \mc{\tilde{S}}_{(\mathfrak{s})}$. We next evaluate the objective over the set ${\tilde{S}}_{(\mathfrak{s})}$. For any $s\in\mc{\tilde{S}}_{(\mathfrak{s})}$, $s_i$ is an $\epsilon_i \in[0, \epsilon]$ perturbation on the nearest sparsity level $s^*_i\in\mc{S}_{(\mathfrak{s})}$, i.e., $s_i = s^*_{i} \pm \epsilon_i$ $\forall$ $i$, where $s$ is such that $\sum_{i=1}^N s_i = \mathfrak{s}$. It can be shown easily that $\|s-s^*\| < \sqrt{N} \epsilon $ $\forall$ $s\in\mc{\tilde{S}}_{(\mathfrak{s})}$. The modified ratio curves $\tilde{r}_i(s_i)$ follow \eqref{rfunc} with $\bar{\epsilon}_i = \nicefrac{\epsilon_i}{\epsilon}\in [0,1]$ $\forall$ $i$. To show that the local minima will only exist in the subset $\mc{S}_{(\mathfrak{s})}$ of $\mc{\tilde{S}}_{(\mathfrak{s})}$, we next prove that $\tilde{F}(s)\geq \tilde{F}(s^*)$ for $s\in\mc{\tilde{S}}_{(\mathfrak{s})}$.
\par Since $\epsilon_i$ follow $0<\epsilon_i <\frac{1}{N}$ $\forall$ $i$, the perturbations are such that $\sum_{i=1}^N \pm \epsilon_i=0$ to ensure $\sum_{i=1}^N s=\mathfrak{s}$. Using \eqref{HHtilde}, $\tilde{F}(s)$ over the set $\mc{\tilde{S}}_{(\mathfrak{s})}$ is evaluated as $\tilde{F}(s)|_{s\in\mc{\tilde{S}}_{(\mathfrak{s})}}$ $=\tilde{F}(s^*) +$  $\Delta \tilde{F}(s)|_{s\in\mc{\tilde{S}}_{(\mathfrak{s})}}$, where $\Delta \tilde{F}(s)|_{s\in\mc{\tilde{S}}_{(\mathfrak{s})}}$ follows the condition:
\begin{align}
 &\Delta \tilde{F}(s)|_{s\in\mc{\tilde{S}}_{(\mathfrak{s})}} > H((D-r_1)\bar{\epsilon}_1,\ldots,(ND-r_N)\bar{\epsilon}_N) \nn \\
 &+\Delta H (s) + \sigma ((D-r_1)\bar{\epsilon}_1 + (ND-r_N)\bar{\epsilon}_N)  0,\label{solution2}
\end{align}
where $\Delta H(s)>0$. Moreover, it can be shown easily that as $\epsilon_i \to 0$ $\forall \ i$, $\Delta \tilde{F}(s)|_{s\in\mc{\tilde{S}}_{(\mathfrak{s})}} \to 0$. Using \eqref{solution} and \eqref{solution2}, there exists a positive $\varepsilon>0$ such that if $\|s-s^*\| < \varepsilon $, $\tilde{F}(s) \geq \tilde{F}(s^*)$ for all $s\in \mathbb{S}_{\mathfrak{s}}$ and $s^*\in\mc{S}_{(\mathfrak{s})}$. This proves that all the local minima belong to $\mc{S}_{(\mathfrak{s})}$.
\end{proof}

\subsection{Proof of Proposition \ref{proposition1}}
\label{subsec:proofofprop1}
\begin{proof}
Let $r_{(j)}=[r_{1_j},\ldots,r_{N_j}]^T$ represent the vector of performance ratios obtained in the $j$-th iteration. Let $F_j=f_j-g_j +\sigma h_j$ represent the objective value for $s_{(j)}$ and $\hat{F}_{j,j-1}=f_j-g_{j-1}-\nabla g^T_{j-1} (r_{(j)} - r_{(j-1)}) +\sigma h_j$ represent the relaxed objective value for $s_{(j)}$ with respect to $s_{(j-1)}$. It follows that $\hat{F}_{j,j-1} - F_{j}$ is given as:
\begin{align}
 &\hat{F}_{j,j-1} - F_{j}= g_{j-1} + \nabla g^T_{j-1} (r_{(j)} - r_{(j-1)}) - g_j\nn\\
 &=\frac{1}{N^2}\left( \left(\sum_{i} r_{i_{j-1}}\right)^2 + 2 \sum_{i} r_{i_{j-1}} \left( \sum_{i} r_{i_{j}} - \sum_{i} r_{i_{j-1}} \right) - \left(\sum_{i} r_{i_{j}})^2  \right) \right)\nn\\
 &=\frac{1}{N^2}\left( \sum_{i} r_{i_{j-1}} \left( \sum_{i} r_{i_{j}} - \sum_{i} r_{i_{j-1}} \right)  -  \sum_{i} r_{i_{j}} \left( \sum_{i} r_{i_{j-1}} - \sum_{i} r_{i_{j}} \right)\right)\nn \\
 &=\frac{1}{N^2} \left( \sum_{i} r_{i_{j-1}} - \sum_{i} r_{i_{j}}  \right)^2 = \delta_{j,j-1}.
\end{align}
 We next derive the convergence proof. Let $y_{j-1} = f_{j-1} - g_{j-1} + \sigma h_{j-1}=F_{j-1}$ for the $(j-1)$-th iteration. Then, $y_{j-1}=F_{j-1}=\hat{F}_{j-1,j-1} \geq \hat{F}_{j,j-1}$ as CCP minimizes $\hat{F}(j,j-1)$ at every iteration $j$. Therefore, one can write $y_{j-1} \geq f_{j} - \hat{g}_{j,j-1} \geq y_{j}$. Thus, a non-increasing sequence $\{y_j\}$ is obtained, which converges to a local minima of $\bs{\mc{O}}_5$. Moreover, since $\hat{F}_{j,j-1}-F_j = \delta_{j,j-1}$, where $\delta_{j,j-1} \geq 0$ $\forall$ $j$, beginning from the $0$-th iteration, we have $\hat{F}_{1,0}=F_1 + \delta_{1,0} ,\cdots,\hat{F}_{k,k-1}=F_k + \delta_{k,k-1}$. Also, since $F_j \geq \hat{F}_{j,j-1}$ for any $j$-th iteration, it follows that $F_0 \geq F_1 + \delta_{1,0},\cdots,F_{k-1} \geq F_k + \delta_{k,k-1}$. Since $F_{j}\geq 0$ $\forall$ $j$, we directly obtain \eqref{propb}.
 \end{proof}
 
\bibliography{main}
\end{document}